\documentclass{llncs}

\usepackage{fullpage}

\usepackage{amsmath}
\usepackage{amsfonts}
\usepackage{amssymb}
\usepackage{todonotes}

\usepackage{floatrow}
\usepackage[utf8]{inputenc}
\usepackage[T1]{fontenc} 

\usepackage{pifont}
\usepackage{fancyvrb}

\usepackage{tikz}
\usetikzlibrary{decorations.pathreplacing}
\usetikzlibrary{decorations.pathmorphing}
\usetikzlibrary{shapes,shapes.symbols,automata,arrows,snakes}
\usetikzlibrary{calc}

\tikzset{p0/.style = {shape = circle, draw, thick, minimum size = 0.7cm}}
\tikzset{p1/.style = {rectangle, minimum size=.7cm, draw, thick}}
\tikzset{>=stealth, shorten >=1pt}
\tikzset{every edge/.style = {thick, ->, draw}}
\tikzset{every loop/.style = {thick, ->, draw}}

\usepackage{booktabs}

\newcommand{\buffer}{\gamma}

\newcommand{\markl}{\hspace{-.5pt}\raisebox{.4pt}{\text{\scalebox{.85}{\ding{220}}}}\hspace{-.5pt}}
\newcommand{\markr}{\hspace{-.5pt}\raisebox{.4pt}{\text{\reflectbox{\scalebox{.85}{\ding{220}}}}}\hspace{-.5pt}}

\SaveVerb{sharp}|#|
\newcommand{\sharpsym}{\ensuremath{\mathop{\protect\raisebox{-0.5pt}{\protect\scalebox{1.15}{\protect\UseVerb{sharp}}}}}}

\newcommand{\inc}{\mathtt{i}}
\newcommand{\eps}{\boldsymbol{\epsilon}}

\newcommand{\initmark}{I}

\newcommand{\bigo}[0]{\mathcal{O}}

\newcommand{\size}[1]{|#1|}

\newcommand{\set}[1]{\{ #1 \}}
\newcommand{\nats}{\mathbb{N}}

\renewcommand{\epsilon}{\varepsilon}

\newcommand{\arena}{\mathcal{A}}
\newcommand{\game}{\mathcal{G}}
\newcommand{\cost}{\mathrm{Cst}}

\newcommand{\col}{\Omega}

\newcommand{\cp}{\mathrm{CostParity}}

\newcommand{\init}{\mathrm{Init}}
\newcommand{\update}{\mathrm{Upd}}

\newcommand{\paritydist}{\mathrm{Cor}}
\newcommand{\streettdist}{\mathrm{StCor}}

\newcommand{\answer}[1]{\mathrm{Ans}({#1})}

\newcommand{\up}{\mathbf{UP}}
\newcommand{\coup}{\mathbf{coUP}}

\newcommand{\exptime}{\mathbf{EXPTIME}}
\newcommand{\ptime}{\mathbf{PTIME}}
\newcommand{\twoexp}{\mathbf{2EXPTIME}}
\newcommand{\pspace}{\mathbf{PSPACE}}
\newcommand{\threeexp}{\mathbf{3EXPTIME}}

\newcommand{\R}{\mathfrak{R}}

\newcommand{\aut}{\mathcal{A}}

\newcommand{\delaygame}[1]{\Gamma\!_{f}(#1)}
\newcommand{\delaygamep}[1]{\Gamma\!_{f'}(#1)}
\newcommand{\delaygamec}[2]{\Gamma\!_{f_{#1}}(#2)}

\newcommand{\SigmaI}{\Sigma_I}
\newcommand{\SigmaO}{\Sigma_O}

\newcommand{\stratO}{\tau_O}
\newcommand{\stratI}{\tau_I}
\newcommand{\p}{P}

\title{Games with Costs and Delays}

\author{Martin Zimmermann\thanks{Supported by the project ``TriCS'' (ZI 1516/1-1) of the German Research Foundation (DFG).}}

\institute{Reactive Systems Group, Saarland University, 66123 Saarbrücken, Germany\\
 \email{zimmermann@react.uni-saarland.de}}

\begin{document}
\maketitle

\begin{abstract} 
We demonstrate the usefulness of adding delay to infinite games with quantitative winning conditions. In a delay game, one of the players may delay her moves to obtain a lookahead on her opponent's moves. We show that determining the winner of delay games with winning conditions given by parity automata with costs is EXPTIME-complete and that exponential bounded lookahead is both sufficient and in general necessary. Thus, although the parity condition with costs is a quantitative extension of the parity condition, our results show that adding costs does not increase the complexity of delay games with parity conditions. 

Furthermore, we study a new phenomenon that appears in quantitative delay games: lookahead can be traded for the quality of winning strategies and vice versa. We determine the extent of this tradeoff. In particular, even the smallest lookahead allows to improve the quality of an optimal strategy from the worst possible value to almost the smallest possible one. Thus, the benefit of introducing lookahead is twofold: not only does it allow the delaying player to win games she would lose without, but lookahead also allows her to improve the quality of her winning strategies in games she wins even without lookahead. 
\end{abstract}

\section{Introduction}
\label{sec_intro}
Infinite games are one of the pillars of logics and automata theory with a plethora of applications, e.g., as solutions for the reactive synthesis problem and for the model-checking of fixed-point logics, and as the foundation of the game-based proof of Rabin's theorem to name a few highlights. 

The study of infinite games in automata theory was initiated by the seminal Büchi-Landweber theorem~\cite{BuechiLandweber69}, which solved Church's controller synthesis problem~\cite{Church63}: the winner of an infinite-duration two-player zero-sum  perfect-information game with $\omega$-regular winning condition can be determined effectively. Furthermore, a finite-state winning strategy, i.e., a strategy that is finitely described by an automaton with output, can be computed effectively. Ever since, this result has been extended along various axes, e.g., the number of players, the type of winning condition, the type of interaction between the players, the informedness of the players, etc. Here, we consider two extensions: first, we allow one player to delay her moves, which gives her a lookahead on her opponent's moves. Second, we consider quantitative winning conditions, i.e., the finitary parity condition and the parity condition with costs, which both strengthen the classical parity condition. 

\paragraph*{Delay Games} The addition of delay to the original setting of Büchi and Landweber was already studied by Hosch and Landweber himself~\cite{HoschLandweber72} shortly after the publication of the original result. It concerns the ability of one player to delay her moves in order to gain a lookahead on her opponent's moves. This allows her to win games she would lose without lookahead, e.g., if her first move depends on the third move of the opponent. Delay games are not only useful to model buffers and transmission of data, but the existence of a continuous uniformization function for a given relation is also characterized by winning a delay game (see~\cite{HoltmannKaiserThomas12} for details). Recently, delay games have been reconsidered by Holtmann et al.~\cite{HoltmannKaiserThomas12} and the first comprehensive study of their properties has been initiated~\cite{FridmanLoedingZimmermann11,KleinZimmermann15b,KleinZimmermann16,KleinZimmermann16b,Zimmermann16}. 

In particular, determining the winner of delay games with winning conditions given by deterministic parity automata is $\exptime$-complete, where hardness already holds for safety conditions~\cite{KleinZimmermann16}. These results have to be contrasted with those for the delay-free setting: determining the winner of parity games is in $\up \cap \coup$~\cite{Jurdzinski98}, while the special case of safety games is in $\ptime$. Additionally, exponential constant lookahead is both sufficient and in general necessary~\cite{KleinZimmermann16}, i.e., unbounded lookahead does not offer an additional advantage for the delaying player in $\omega$-regular delay games. Either she wins by delaying an exponential number of moves at the beginning of a play and then never again, or not at all.

\paragraph*{Parity Games with Costs} Another important generalization of the Büchi-Landweber theorem concerns quantitative winning conditions, as qualitative ones are often too weak to capture the specifications of a system to be synthesized. Prominent examples of quantitative specification languages are parameterized extensions of Linear Temporal Logic (LTL) like Prompt-LTL~\cite{KupfermanPitermanVardi09}, which adds the prompt-eventually operator whose scope is bounded in time, and the finitary parity condition~\cite{ChatterjeeHenzingerHorn09}. The latter strengthens the classical parity condition by requiring a fixed, but arbitrary, bound~$b$ on the distance between occurrences of odd colors and the next larger even color. Recently, a further generalization has been introduced in the setting of arenas with costs, i.e., non-negative edge weights. Here, the bound~$b$ on the distance is replaced by a bound on the cost incurred between occurrences of odd colors and the next larger even color~\cite{FijalkowZimmermann14}. This condition, the parity condition with costs, subsumes both the classical parity condition and the finitary parity condition. 

Although the finitary parity condition is a strengthening of the classical parity condition, the winner of a finitary parity game can be determined in polynomial time~\cite{ChatterjeeHenzingerHorn09}, i.e., the solution problem is simpler than that of parity games (unless $\up \cap \coup = \ptime $). In contrast, parity games with costs are not harder than they have to be, i.e., as hard as parity games~\cite{MogaveroMS15}. Furthermore, such games induce an optimization problem: determine the smallest bound~$b$ that allows to satisfy the condition with respect to $b$. It turns out that determining the optimal bound is much harder~\cite{WeinertZimmermann16}: determining whether a given parity game with costs can be won with respect to a given bound~$b$ is $\pspace$-complete, where hardness already holds for the special case of finitary parity games. The same phenomenon manifests itself in the memory requirements of winning strategies for finitary parity games: playing optimally requires exponentially more memory than just winning. 

\paragraph*{Quantitative Delay Games} 

Recently, both extensions have been investigated simultaneously, first in the form of delay games with WMSO$+$U winning conditions, i.e., weak monadic second-order logic with the unbounding quantifier~\cite{Bojanczyk11}. This logic turned out to be too strong: in general, unbounded lookahead is necessary to win such games and the problem of determining the winner has only partially been resolved~\cite{Zimmermann16}. Thus, the search for tractable classes of quantitative winning conditions started. A first encouraging candidate was found in Prompt-LTL: determining the winner of Prompt-LTL delay games is $\threeexp$-complete and triply-exponential constant lookahead is sufficient and in general necessary~\cite{KleinZimmermann16b}. This comparatively high complexity has to be contrasted with that of delay-free Prompt-LTL games, which are already $\twoexp$-complete~\cite{PnueliRosner89a}. Thus, adding delay incurs an exponential blowup, which is in line with other results for delay games mentioned above. 

\paragraph*{Our Contribution} 

In this work, we investigate delay games with finitary parity conditions and parity conditions with costs and demonstrate the positive effects of adding lookahead to quantitative games. As both the finitary parity condition and the parity condition with costs subsume the classical safety condition, we immediately obtain $\exptime$-hardness and an exponential lower bound on the necessary lookahead. 

First, we show that the exponential lower bound on the lookahead is tight, by presenting a matching  upper bound via a pumping argument for the opponent of the delaying player: if he wins with a large enough lookahead, then he can pump his moves to win for an arbitrarily large lookahead. 

Second, using this result, we construct a \emph{delay-free} parity game with costs that is equivalent to the delay game with fixed exponential lookahead. Furthermore, we show that the resulting game can be solved in exponential time, i.e., we prove $\exptime$-completeness of determining the winner of delay games with finitary parity conditions and parity conditions with costs. 

These two results show that adding costs to delay games with parity conditions comes for free, i.e., the complexity does not increase, both in terms of necessary lookahead and the computational complexity of determining the winner. These results are similar to the ones for Prompt-LTL delay games, which are just as hard as LTL delay games. However, unlike the latter games, which are $\threeexp$-complete, delay games with quantitative extensions of parity conditions are only $\exptime$-complete.

Third, we investigate the power of delay in quantitative games: having a lookahead on the moves of her opponent allows a player not only to win games she would lose without this advantange, but also to improve the (semantic) quality of her winning strategies in games she wins even without lookahead. For example, we present a game induced by an automaton with $\bigo(n)$ states such that she wins the delay-free game with bound~$n$ on every play, which is close to the worst-case~\cite{FijalkowZimmermann14}. However, with a lookahead of just a single move, she wins the game with bound~$1$. Thus, lookahead can be traded for quality and vice versa. In further examples, we show that this tradeoff can be gradual, i.e., decrementing the bound requires one additional move lookahead, and that exponential lookahead may be necessary to obtain the optimal bound. Furthermore, we present matching upper bounds on the tradeoff between lookahead and quality. 

All examples we present in this work are from a very small fragment, finitary Büchi games, i.e., every edge has cost~$1$ and the only colors are $1$ and $2$. Thus, every occurrence of a $1$ has to be answered quickly by an occurrence of a $2$. Our results show that even these games are as hard as general parity games with costs, i.e., games with arbitrary weights and an arbitrary number of colors. 

Fourth, we consider the more general setting of delay games with winning conditions given by Streett automata with costs. Streett conditions are obtained from parity conditions by abandoning the hierarchy between requests (even numbers) and responses (odd numbers) by allowing arbitrary sets of states to represent requests and responses. We prove that the techniques developed for parity conditions can be extended to Streett conditions with costs, at the cost of an exponential blowup: doubly-exponential constant lookahead is sufficient and such games can be solved in doubly-exponential time. Also, we show that an \emph{optimal} winning strategy requires in general doubly-exponential lookahead. Whether general strategies do as well is an open problem. Note that this is already the case for delay games with qualitative Streett conditions.

\section{Definitions}
\label{sec_defs}
We denote the non-negative integers by $\nats$. Given two infinite words $\alpha \in \SigmaI^\omega$ and $\beta  \in \SigmaO^\omega$ we write ${ \alpha \choose \beta}$ for the word ${\alpha(0) \choose \beta(0) } {\alpha(1) \choose \beta(1) }  {\alpha(2) \choose \beta(2) } \cdots  \in (\SigmaI \times \SigmaO)^\omega$. Similarly, we write ${x \choose y}$ for finite words $x$ and $y$, provided they are of equal length.

\subsection{Parity Automata with Costs}
\label{subsec_automata}
A parity automaton with costs~$\aut = (Q, \Sigma, q_\initmark, \delta, \col, \cost)$ consists of a finite set $Q$ of states, an alphabet~$\Sigma$, an initial state~$q_\initmark \in Q$, a (deterministic and complete) transition function~$\delta \colon Q \times \Sigma \rightarrow Q$, a coloring~$\col \colon Q \rightarrow \nats$, and a cost function~$\cost \colon \delta \rightarrow \set{\eps, \inc}$ (here, and whenever convenient, we treat $\delta$ as a relation). If $\cost(q, a, q') = \eps$, then we speak of an $\epsilon$-transition, otherwise of an increment-transition. The size of $\aut$ is defined as $\size{\aut} = \size{Q}$.

The run of $\aut$ starting in $q_0 \in Q$ on a finite word~$w = a_0 \cdots a_{n-1}$ is the unique sequence
\begin{equation}
(q_0, a_0, q_1)\, (q_1, a_1, q_2) \, \cdots \, (q_{n-1}, a_{n-1}, q_n) \in \delta^*, \label{eq_run} \end{equation}
i.e., $q_{j+1} = \delta(q_j,a_j)$ for every $j < n$.
We say the run ends in $q_n$ and define its cost as the number of traversed increment-transitions. The run of $\aut$ starting in $q_0$ on an infinite word and its cost is defined analogously. If we speak of \emph{the} run of $\aut$ on an infinite word, then we mean the one starting in $q_\initmark$. 
 
As usual in parity games, we interpret the occurrence of an odd color~$c$ as a request, which has to be answered by an even $c' > c$. Hence, we define $\answer{c} = \set{c' \in \col(Q) \mid c' > c \text{ and $c'$ is even} }$. An infinite run $\rho = (q_0, a_0, q_1)\, (q_1, a_1, q_2) \, (q_{2}, a_{2}, q_3) \, \cdots $
is accepting, if it satisfies the parity condition with costs: there is a bound~$b \in \nats$ such that for almost all $n$ with odd~$\col(q_n)$, the cost between the positions~$n$ and $n'$ is at most $b$, where $n'$ is minimal such that $\col(q_{n'})$ answers the request~${\col(q_n)}$. Formally, we require the cost~$\limsup_{n\rightarrow \infty} \paritydist(\rho,n) $ of $\rho$ to be finite,
where the cost-of-response~$\paritydist(\rho, n)$ is defined to be $0$ if $\col(q_n)$ is even and defined to be
	\begin{align*}
\min\set{ \cost((q_n, a_n, q_{n+1})\cdots (q_{n'-1}, a_{n'-1}, q_{n'})) \mid
  n' > n \text{ and } \col(q_{n'}) \in \answer{\col(q_n)}}
	\end{align*}
if $\col(q_n)$ is odd,
where $\min \emptyset = \infty$. 

The language of $\aut$, denoted by $L(\aut)$, contains all infinite words whose run of $\aut$ is accepting. If every transition in $\aut$ is an $\epsilon$-transition, then the parity condition with costs boils down to the parity condition, i.e., $\aut$ is a classical parity automaton and $L(\aut)$ is $\omega$-regular. In contrast, if every transition is an increment-transition, then $\aut$ is a finitary parity automaton, which have been studied by Chatterjee and Fijalkow~\cite{ChatterjeeFijalkow11}. 

\begin{remark}
In the following, and without mentioning it again, we only consider parity automata with both even and odd colors. If this is not the case, then the automaton is trivial, i.e., its language is either universal or empty.
\end{remark}

In our examples, we often use finitary Büchi automata, i.e., automata having only increment-transitions and colors~$1$ and $2$. Hence, every state of color~$2$ answers all pending requests and all other states open a request (of color~$1$). Hence, a run of such an automaton is accepting, if there is a bound~$b$ such that every infix of length~$b$ contains a transition whose target state has color~$2$. Furthermore, the cost of the run is the smallest~$b'$ such that almost every infix of length~$b'$ contains a transition whose target state has color~$2$. In figures, we denote states of color~$2$ by doubly-lined states, all other states have color~$1$, and we do not depict the cost function, as it is trivial.

\begin{example}
\label{example_automata}
Consider the finitary Büchi automaton~$\aut_n$, for $n>1$,  on the left side of Figure~\ref{fig_automataexample} over the alphabet~$\SigmaI \times \SigmaO = \set{1, \ldots, n, \#}^2$. A word~${\alpha(0) \cdots \alpha(m) \choose \beta(0) \cdots \beta(m)}$ leads~$\aut_n$ from the initial state, which opens a request, to the rightmost state, which answers all requests, without visiting the latter in the meantime, if, and only if, it has the following form: $\alpha(i) \neq \sharpsym$ for all $i$, $\beta(i) = \sharpsym$ if, and only if $i = m$, and $\alpha(1) \cdots (m)$ has two occurrences of $\beta(0)$ with no larger letter in between (a so-called bad $j$-pair for $j = \beta(0)$), where the second occurrence is $\alpha(m)$. Call such a word productive.

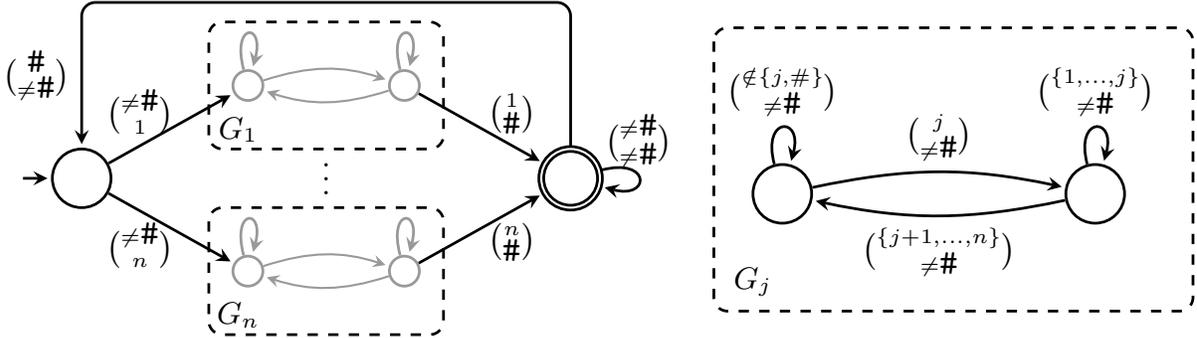
\begin{figure*}
\centering
\scalebox{1.3}{
\begin{tikzpicture}

\node at (6.3,0) {

\begin{tikzpicture}[thick]
\tikzset{state/.style = {shape = circle, draw, thick, minimum size = 0.6cm, inner sep=0}}

\node[state] 				(lj)	 at (0,0)	{};
\node[state] 				(rj)	 at (3.2,0)	{};

\node 								 at (-0.3, -.9) {$G_j$};
 
\path[-stealth] 
(lj) edge[loop above] node{${\notin \set{j,\#} \choose \neq \sharpsym}$} ()
(rj) edge[loop above] node{${\set{1, \ldots, j} \choose  \neq \sharpsym}$} ()
(lj) edge[bend left=12] node[above]{${j \choose  \neq \sharpsym}$} (rj)
(rj) edge[bend left=12] node[below]{${ \set{j+1, \ldots, n} \choose  \neq \sharpsym}$} (lj)
;  
			
\draw[rounded corners, dashed] (-.7,1.7) rectangle (4.2,-1.2)
;

\end{tikzpicture}
};

\node at (0,0) {
\begin{tikzpicture}[thick]
\tikzset{state/.style = {shape = circle, draw, thick, minimum size = 0.6cm}}

\node[state] 				(init)	 at (-.5,0) 	{};
\node[state,accepting] 	(final) at (4.5,0) {};

\path[->] (-1.1,0) edge (init);

\draw[rounded corners, dashed] (.8,.3) rectangle (3.2,1.6);
\node  at (2, .1) {$\vdots$};
\draw[rounded corners, dashed] (.8,-.3) rectangle (3.2,-1.6);
\node 								 at (1.1, .5) {\footnotesize $G_1$};
\node 								 at (1.1, -1.4) {\footnotesize $G_n$};

\path (final) edge[loop right] node[above] {$ {\neq \sharpsym \choose \neq \sharpsym} $} ();

\tikzset{state/.style = {shape = circle, draw, thick, minimum size = 0.08cm}}

\node[state, black!40] (l1) at (1.2,.95) {}; 
\node[state, black!40] (r1) at (2.8,.95) {}; 

\path[black!40]
(l1) edge[semithick, bend left =18] (r1)
(r1) edge[semithick, bend left =18] (l1)
(l1) edge[semithick, loop above] ()
(r1) edge[semithick, loop above] ()
;

\node[state, black!40] (ln) at (1.2,-.95) {}; 
\node[state, black!40] (rn) at (2.8,-.95) {}; 

\path[black!40]
(ln) edge[semithick, bend left =18] (rn)
(rn) edge[semithick, bend left =18] (ln)
(ln) edge[semithick, loop above] ()
(rn) edge[semithick, loop above] ()
;

\path
(init) edge node[above, near start] {${\neq\sharpsym \choose 1}$} (l1)
(init) edge node[below, near start] {${\neq\sharpsym \choose n}$} (ln)
(r1) edge  node[above, near end] {${1 \choose \sharpsym}$} (final)
(rn) edge  node[below, near end] {${n \choose \sharpsym}$}(final);

\path[draw, rounded corners,->] (final.north) -- (4.5,1.8)  -- (-.5,1.8) -- node[left] {${\sharpsym \choose \neq \sharpsym}$} (init.north);

\end{tikzpicture}
};
\end{tikzpicture}}
\caption{The automaton~$\aut_n$ (on the left), which contains gadgets~$G_1 , \ldots, G_n$ (on the right). Transitions not depicted are defined as follows: those of the form~${ \sharpsym \choose b }$ for arbitray~$b \in \SigmaO$ lead to an (accepting) sink of color~$2$, those of the form~${ \neq \sharpsym \choose \sharpsym }$ lead to a (rejecting) sink of color~$1$. Both sinks are not drawn.}
\label{fig_automataexample}
\end{figure*}

Then, the language of $\aut_n$ contains two types of words: first all those that can be decomposed into $w_0 x_0 { \sharpsym \choose \neq \sharpsym } w_1 x_1 { \sharpsym \choose \neq \sharpsym } w_2 x_2 { \sharpsym \choose \neq \sharpsym } \cdots$ where each $w_i$ is productive, where $\sup_i \size{w_i} < \infty$, and where each $x_i$ is in $(\set{1, \ldots, n}^2)^*$. The other way $\aut_n$ accepts a word is by reaching the accepting sink state, which it does for all words of the form~$w_0 x_0 { \sharpsym \choose \neq \sharpsym } \cdots w_{m-1} x_{m-1} { \sharpsym \choose \neq \sharpsym } w_m' {\sharpsym \choose *} (\set{1, \ldots, n}^2)^\omega$ where each $w_i$ is productive, each $x_i$ is in $(\set{1, \ldots, n}^2)^*$, $w_m'$ is a strict prefix of a productive word, and $*$ denotes an arbitrary letter.
\end{example}

To conclude this introductory subsection, we state a replacement lemma for accepting runs:  fix a parity automaton with costs~$\aut = (Q, \Sigma, q_\initmark, \delta, \col, \cost)$ and let $(q_0, a_0, q_1) \cdots  (q_{n-1}, a_{n-1}, q_n)$ be a non-empty finite run of $\aut$. Its type is the tuple~$(q_0, q_n, c_0, c_1, \ell)$  where $c_0, c_1 \in \col(Q) \cup \set{\bot}$ and $\ell \in \set{\eps, \inc}$ are defined (using $\max \emptyset = \bot$) as 
\begin{itemize}
	\item $c_0 = \max \set{\col(q_j) \mid 0 \le j \le n \text{ and } \col(q_j) \text{ is even}}$,
	\item $c_1 = \max\set{ \col(q_j) \mid 0 \le j \le n \text{, } \col(q_j) \text{ is odd, and } \col(q_{j'}) \notin \answer{\col(q_j)} \text{ for all } j < j' \le n}$, i.e., $c_1$ is the maximal unanswered request, and
	\item $\ell = \inc$ if, and only if, the run contains an increment-transition.
\end{itemize}

The following lemma shows that we can replace infixes of accepting runs by infixes of the same type, provided the replacements are of bounded length.

\begin{lemma}
\label{lemma_runreplacement}
Let $(\rho_j^1)_{j \in \nats}$ and $(\rho_j^2)_{j \in \nats}$ be sequences of finite runs of a parity automaton with costs such that $\rho_j^1$ and $\rho_j^2$ have the same type for every $j$, and such that $\sup_j \size{\rho_j^2} = d < \infty$. Let $\rho' = \rho_0^1 \rho_1^1 \rho_2^1 \cdots $ and $\rho'' = \rho_0^2  \rho_1^2 \rho_2^2 \cdots$.

If $\limsup_{n\rightarrow \infty} \paritydist(\rho',n) \le b$ for some $b$, then $\limsup_{n\rightarrow \infty} \paritydist(\rho'',n) \le (b+2)\cdot d $, i.e., if  $\rho'$ is an accepting run, then $\rho''$ is an accepting run as well.
\end{lemma}

\begin{proof}
Let $j_0$ be such that every request in the suffix~$\rho'_s = \rho^1_{j_0} \rho^1_{j_0+1} \rho^1_{j_0+2} \cdots $ of $\rho'$ is answered with cost at most $b$. Consider a position~$n$ in the corresponding suffix~$\rho''_s = \rho^2_{j_0} \rho^2_{j_0+1} \rho^2_{j_0+2} \cdots$ of $\rho''$ with odd color, i.e., a request in $\rho''$. We show that this request is answered with cost at most $(b+2) \cdot d$ by first identifying a corresponding request in $\rho_s'$.
 If it is already answered in $\rho_j^2$, then with cost at most $d \le (b+2) \cdot d$, as $\rho_j$ has at most $d$ increment transitions. Thus, assume it is not. Then, by $\rho_j^1$ and $\rho_j^2$ having the same type, there is a (potentially) larger request in $\rho_j^1$, which is not answered in $\rho_j^1$, but is answered with cost at most $b$. 
 
Furthermore, by the same reasoning as above, this answer (or an even larger one) also appears in $\rho_s''$. Also, each of the $b$ increments in $\rho_s'$ is (in the worst case) replaced by $d$ increments when going from $\rho_s'$ to $\rho_s''$. Finally, we have to account for an additional $2d$ increments, if both the request and the response are in some part $\rho_j^2$, i.e., they are replaced.. Altogether, the request at position~$n$ is answered with cost at most $(b+2)\cdot d$. Thus, in $\rho''$ almost all requests are answered with that cost. 
\end{proof}

The special case where both the lengths of the $\rho_j^1$ and the lengths of the $\rho_j^2$ are bounded is useful later on, too.

\begin{corollary}
\label{corollary_runreplacement}
Let $(\rho_j^1)_{j \in \nats}$ and $(\rho_j^2)_{j \in \nats}$ be sequences of finite runs of a parity automaton with costs such that $\rho_j^1$ and $\rho_j^2$ have the same type for every $j$, and such that $\sup_j \size{\rho_j^1} < \infty$ and $\sup_j \size{\rho_j^2} < \infty$.

Then,  $\rho_0^1 \rho_1^1  \rho_2^1 \cdots $ is an accepting run if, and only if, $ \rho_0^2  \rho_1^2 \rho_2^2 \cdots$ is an accepting run.
\end{corollary}

Finally, let us remark that the type of a run can be computed on the fly. Let $T_\aut = Q^2 \times (\col(Q) \cup\set{\bot})^2 \times \set{\eps, \inc}$ be the set of types of $\aut$. Define $\init_\aut(q) = (q,q,c_0, c_1, \eps)$ with 
\begin{itemize}
	\item $c_0 = \col(q)$, if $\col(q)$ is even, and $c_0  = \bot$ otherwise, and 
	\item $c_1 = \col(q)$, if $\col(q)$ is odd, and $c_1 = \bot$ otherwise.
\end{itemize}
Thus, $\init_\aut(q)$ can be understood as the type of the empty prefix of a run starting in state~$q$.
Furthermore, define $\update_\aut((q_0, q_1, c_0, c_1, \ell), a) =  (q_0, q_1', c_0', c_1', \ell')$ where 
\begin{itemize}
	\item $q_1' = \delta(q_1, a)$,
	\item $c_0' = \begin{cases}
		\max\set{c_0, \col(q_1')} &\text{ if $\col(q_1)$ is even,}\\
		c_0 &\text{ if $\col(q_1)$ is odd,}
	\end{cases}$
	\item $c_1' = \begin{cases}
 \max \set{c_1, \col(q_1')} &\text{if $\col(q_1')$ is odd,}\\
  \bot &\text{if $\col(q_1') \in \answer{ c_1}$,}\\
 c_1 &\text{else,}
 \end{cases}$ \quad and
	\item $\ell' = \inc$ if $\ell = \inc$ or if $(q_1, a, q_1')$ is an increment transition.  
\end{itemize}
Here, we use $\max \set{\bot, c} = c$ for every $c \in \col(Q)$. 

\begin{remark}
\label{remark_tracking}
The run~$(q, a, \delta(q,a))$ has type~$\update_\aut(\init_\aut(q), a)$, and for every non-empty finite run~$\rho = (q_0, a_0, q_1) \cdots  (q_{n-1}, a_{n-1}, q_n)$ and every $a \in \Sigma$: if $\rho$ has type~$t$, then $\rho \cdot (q_n, a ,\delta(q_n,a)) $ has type~$\update_\aut(t, a)$.
\end{remark}

\subsection{Delay Games}
\label{subsec_delaygames}
A delay function is a map~$f \colon \nats \rightarrow \nats\setminus\set{0}$, which is said to be constant, if $f(i)=1$ for every $i>0$. A delay game~$\delaygame{L}$ consists of a delay function~$f$ and a winning condition~$L \subseteq (\SigmaI \times \SigmaO)^\omega$. It is played in rounds~$i =0,1,2, \ldots$ between Player~$I$ and Player~$O$. In round~$i$, first Player~$I$ picks $u_i \in \SigmaI^{f(i)}$, then Player~$O$ answers by picking $v_i \in \SigmaO$. Player~$O$ wins the resulting play $(u_0, v_0) (u_1, v_1) (u_2, v_2) \cdots$ if the outcome~${ u_0 u_1 u_2 \cdots \choose v_0 v_1 v_2 \cdots}$ is in $L$, else Player~$I$~wins.

A strategy for Player~$I$ in $\delaygame{L}$ is a mapping~$\stratI \colon \SigmaO^* \rightarrow \SigmaI^*$ such that $\size{\stratI(w)} = f(\size{w})$; a strategy for Player~$O$ is a mapping~$\stratO \colon \SigmaI^*\rightarrow \SigmaO$. A play~$(u_0, v_0) (u_1, v_1) (u_2, v_2) \cdots$ is consistent with $\stratI$, if $u_i = \stratI(v_0 \cdots v_{i-1})$ for every $i$ and it is consistent with $\stratO$, if $v_i = \stratO(u_0 \cdots u_i)$ for every $i$. A strategy for Player~$\p \in \set{I,O}$ is winning, if the outcome of every consistent play is winning for Player~$\p$. In this case we say that Player~$\p$ wins $\delaygame{L}$. A game is determined if one of the players wins it.

Parity automata with costs recognize Borel languages~\cite{FijalkowZimmermann14}.\footnote{The result proven there is about winning conditions of games, which covers languages of automata as a special case.} Thus, the Borel determinacy result for delay games~\cite{KleinZimmermann15b} is applicable: delay games with winning conditions given by parity automata with costs are determined.

In this work, we study delay games whose winning conditions are given by a parity automaton with costs~$\aut$. In particular, we determine the complexity of determining whether Player~$O$ wins $\delaygame{L(\aut)}$ for some $f$ and what kind of $f$ is necessary in general. Furthermore, we study the tradeoff between the quality and the necessary lookahead of winning strategies. 

To this end, given a winning strategy~$\stratO$ for Player~$O$ in a delay game~$\delaygame{L(\aut)}$, let \[\cost_\aut(\stratO) = \sup_w ( \limsup_{n\to\infty}\paritydist(\rho(w), n)),\] where $w$ ranges over the outcomes of $\stratO$ and where $\rho(w)$ denotes the unique run of $\aut$ on $w$. A winning strategy for $\delaygame{L(\aut)}$ is optimal, if it has minimal cost among the winning strategies for $\delaygame{L(\aut)}$. Note that the delay function~$f$ is fixed here, i.e., the cost of optimal strategies might depend on $f$. Furthermore, the cost of a winning strategy might be infinite, if has outcomes with arbitrarily large cost. This is possible, as the parity condition with costs only requires a bound for every run, but not a uniform one among all runs. 

\begin{example}
\label{example_delaygame}
Consider delay games with winning conditions~$L_n = L(\aut_n)$, where $\aut_n$ is the parity automaton with costs from Example~\ref{example_automata}. 

We claim that Player~$O$ wins $\delaygame{L_n}$ if $f(0) \ge 2^n+1$, based on the following fact: every word over~$\set{1, \ldots, n}$ of length~$2^n$ contains a bad $j$-pair for some $j$~\cite{KleinZimmermann16}.
	
Now, consider the situation after the first move of Player~$I$: $\aut_n$ is in the initial state, as Player~$O$ has not yet produced any letters, and Player~$I$ has picked at least $2^n+1$ letters~$\alpha(0) \cdots \alpha(f(0)-1)$ lookahead. Let $w = \alpha(1) \cdots \alpha(2^n)$. We consider several cases:

If $w$ contains a $\sharpsym$ without a bad $j$-pair in the prefix before the first $\sharpsym$, then Player~$O$ can pick $1$'s until the automaton has reached the accepting sink state. The resulting run is accepting with cost~$0$.

Otherwise, $w$ contains a bad $j$-pair without an earlier occurrence of $\sharpsym$, for some $j$. Then, Player~$0$ picks $\beta(0) = j$ and afterwards arbitrary letters (but $\sharpsym$) until the second~$j$ constituting the bad $j$-pair, which she answers by a $\sharpsym$. The resulting word is productive, i.e., it leads the automaton to the rightmost state. Then, either the run stays in this state forever, which implies it is accepting, or the initial state is reached again. Then, Player~$O$ can iterate her strategy, as the size of the lookahead is non-decreasing.

 Hence, if the initial state is visited infinitely often, then the rightmost state of color~$2$ is visited at least once in every run infix of length~$2^n+1$, which implies that the resulting outcome is winning for Player~$O$. Every other outcome has only finitely many requests and is therefore winning with cost~$0$.  In particular, the cost of this strategy is bounded from above by $2^n+1$.
	
On the other hand, Player~$I$ wins $\delaygame{L_n}$ if $f(0) < 2^n+1$. Let $w_n \in \SigmaI$ be inductively defined by $w_1 = 1$ and $w_{j} = w_{j-1} \cdot  j \cdot w_{j-1} $ for $j>0$. Then, $w_n$ has length $2^n-1$ and contains no bad $j$-pair for every $j$, and trivially no $\#$. 
	
	Let Player~$I$ pick the prefix of length~$f(0)$ of $1w_n$ in the first round and let Player~$O$ answer by picking $\beta(0)$. If $\beta(0) = \#$, then $\aut_n$ reaches the rejecting sink-state of color~$1$. If $\beta(0) = j$, then Player~$I$ picks~$j'$ ad infinitum for some $j' \notin \set{j,\sharpsym}$. Then, $\aut_n$ neither reaches the accepting sink state (since he never plays a $\sharpsym$) nor the rightmost state (as he has not produced a bad $j$-pair). Hence, the resulting run is rejecting.
\end{example}

One can even show an exponential lower bound on the cost of an optimal winning strategy for Player~$O$: using the word~$w_n$ to begin each round and then restart the play by picking~$\sharpsym$ after Player~$O$ has reached the rightmost state by picking $\sharpsym$, he can enforce a cost-of-response of $2^n+1$ for infinitely many requests, i.e., a winning strategy for Player~$O$ has at least cost~$2^n+1$ as well. 

\begin{proposition}
\label{prop_lowerbounds}
For every $n$, there is a finitary parity automaton~$\aut_n$ of size~$\bigo(n)$ such that Player~$O$ wins $\delaygame{L(\aut_n)}$ if, and only if, $f(0) \ge 2^n+1$. Furthermore, the cost of an optimal winning strategy for $\delaygame{L(\aut_n)}$ is $2^n+1$, for every such $f$.
\end{proposition}

\section{Constant Lookahead Suffices}
\label{sec_constant}
In this section, we prove that constant delay functions suffice for Player~$O$ to win delay games with winning conditions specified by parity automata with costs and give an exponential upper bound on the necessary initial lookahead. To this end, we generalize a similar result proven for delay games with winning conditions given by parity automata~\cite{KleinZimmermann16}. This is the first step towards proving that the winner of such a game can be determined effectively. 

\begin{theorem}
\label{thm_constantsuffices}
Let $L$ be recognized by a parity automaton with costs~$\aut$ with $n$ states and $k$ colors, and let $f$ be the constant delay function with $f(0) = 2^{2n^4k^2+1}$. The following are equivalent:
\begin{enumerate}

	\item Player~$I$ wins $\delaygame{L(\aut)}$.

	\item Player~$I$ wins $\delaygamep{L(\aut)}$ for every delay function~$f'$.

\end{enumerate}
\end{theorem}

\begin{proof}
We only prove the non-trivial direction by taking a winning strategy for $\delaygame{L(\aut)}$ and pumping its moves to obtain a winning strategy  for $\delaygamep{L(\aut)}$, which might require Player~$I$ to provide more lookahead. 

To this end, let $\aut = (Q, \SigmaI \times \SigmaO, q_\initmark, \delta, \col, \cost)$, let $T_\aut$ be the set of types of $\aut$, and recall that $\init_\aut$ and $\update_\aut$ compute the type of a run as described in Remark~\ref{remark_tracking}. We extend $\aut$ so that it tracks the type of its runs using the state set~$Q \times T_\aut$. This information does not change the language of the automaton, but is useful when pumping the moves in $\delaygame{L(\aut)}$.

Formally, we define $\aut' = (Q', \SigmaI \times \SigmaO, q_\initmark', \delta', \col', \cost')$ with
\begin{itemize}
	\item $Q' = Q \times T_\aut$,
	\item $q_\initmark' = (q_\initmark, \init(q_\initmark))$ ,
	\item $\delta'((q, t),{a \choose b}) = (\delta(q,{a \choose b}),\update(t,{a \choose b}))$,
	\item $\col'(q,t) = \col(q)$, and
	\item $\cost'((q,t),{a \choose b},(q',t')) = \col(q,{a \choose b},q')$.
\end{itemize}

Now, define $\delta_P \colon 2^{Q'} \times \SigmaI \rightarrow 2^{Q'}$ via
\[
\delta_P (S, a) = \left\{\left.\delta'\left( (q,t), {a \choose b} \right) \right| (q,t) \in S \text{ and } b \in \SigmaO\right\}.\]
Intuitively, $\delta_P$ is obtained as follows: take $\aut'$, project away $\SigmaO$, and apply the power set construction (ignoring the coloring and the costs). Then, $\delta_P$ is the transition function of the resulting deterministic automaton.
As usual, we extend $\delta_P$ to $\delta_P^+ \colon 2^{Q'} \times \SigmaI^+ \rightarrow 2^{Q'}$ via $\delta_P^+(S,a) = \delta_P(S,a)$ and $\delta_P^+(S, wa) = \delta_P( \delta_P^+(S, w), a)$. 

\begin{remark}
\label{remark_powersetcharac}
The following are equivalent for $q \in Q$ and $w \in \SigmaI^+$:
\begin{enumerate}
	\item $(q',t') \in \delta_P^+(\set{(q, \init_\aut(q))},w)$. 
	\item There is a $w' \in (\SigmaI \times \SigmaO)^+ $ whose projection to $\SigmaI$ is $w$ such that the run of $\aut$ processing $w'$ from $q$ ends in $q'$ and has type~$t'$.
\end{enumerate}
\end{remark}

Let $\stratI$ be a winning strategy for Player~$I$ in $\delaygame{L(\aut)}$ and let $f'$ be an arbitrary delay function. We construct a winning strategy~$\stratI'$ for Player~$I$ in $\delaygamep{L(\aut)}$ by simulating a play of $\delaygamep{L(\aut)}$ by a play in $\delaygame{L(\aut)}$. For the sake of brevity, we denote $\delaygame{L(\aut)}$ by $\Gamma$ and $\delaygamep{L(\aut)}$ by $\Gamma'$ from now on.

Recall that we only consider automata~$\aut$ with non-trivial colorings. Thus, we can bound the number of $\aut$'s types by $2n^2k^2$, where $n$ is the number of states and $k$ the number of colors. 
 We define $d =  \size{(2^{Q'})^Q} = 2^{2n^4k^2}$. In the simulating play in $\Gamma$, the players make their moves in blocks of length~$d$: Player~$I$'s are denoted by $\overline{a_i}$ and Player~$O$'s by $\overline{b_i}$, i.e., in the following, every $\overline{a_i}$ is in $\SigmaI^d$ and every $\overline{b_i}$ is in $\SigmaO^d$. Furthermore, we say that a decomposition $\overline{a_i} = xyz$ is \emph{pumpable}, if $y$ is non-empty and if 
 \[
	\delta_P^+(\set{(q, \init_\aut(q))},x) = \delta_P^+(\set{(q, \init_\aut(q))},x y)
\]
for every $q \in Q$.  As $\aut$ is complete, $\delta_P^+(\set{(q, \init_\aut(q)),w})$ is always non-empty, which implies the following remark.

\begin{remark}
\label{rem_pumpdecomp}
Every $\overline{a} \in \SigmaI^d$ has a pumpable decomposition. 
\end{remark}

Now, we begin the construction of $\stratI'$. Note that we have $f(0) = 2d$. Thus, let $\stratI(\epsilon) = \overline{a_0}\overline{a_1}$ be the first move of Player~$I$ in $\Gamma$ according to $\stratI$. Remark~\ref{rem_pumpdecomp}  yields pumpable decompositions~$\overline{a_0} = x_0 y_0 z_0$ and $\overline{a_1} = x_1 y_1 z_1$. We pick $h_0>0$ such that $\size{x_0 (y_0)^{h_0} z_0} \ge f'(0)$ and define $\alpha(0) \cdots \alpha(\ell_0-1) = x_0 (y_0)^{h_0} z_0$. Similarly, we pick $h_1 >0$ such that $\size{x_1 (y_1)^{h_1} z_1} \ge \sum_{j = 1}^{\ell_0-1}f'(j)$ and define $\alpha(\ell_0) \cdots \alpha(\ell_1-1) = x_1 (y_1)^{h_1} z_1$. Now, we define the strategy~$\stratI'$ for Player~$I$ in $\Gamma'$ to pick the prefix of length~$\sum_{j=0}^{\ell_0-1}f'(j)$ of $\alpha(0) \cdots \alpha(\ell_1-1)$ during the first $\ell_0$ rounds, independently of the choices of Player~$O$.
This prefix is well-defined by the choices of $h_0$ and $h_1$. The remaining letters of $\alpha(0) \cdots \alpha(\ell_1-1)$ are stored in a buffer~$\buffer_1$.  During these first $\ell_0$ rounds, Player~$O$ answers by producing some $\beta(0) \cdots \beta(\ell_0-1)$. 

Thus, we are in the following situation for $i=1$ (see the solid part of Figure~\ref{fig_equivalenceproof}).
\begin{itemize}

	\item In $\Gamma$, Player~$I$ has picked $\overline{a_0}, \ldots,\overline{a_i}$ such that for every $j \le i$: $x_j y_j z_j$ is a pumpable decomposition of $\overline{a_j}$. Furthermore, Player~$O$ has picked $\overline{b_0}, \ldots,\overline{b_{i-2}}$.

	\item In $\Gamma'$, Player~$I$ has picked the prefix of length~$\sum_{j=0}^{\ell_{i-1}-1}f'(j) $ of
	\[
	\alpha(0) \cdots \alpha(\ell_i -1) = x_0 (y_0)^{h_0} z_0 \cdots x_i (y_i)^{h_i} z_i 
	\]
	while the remaining suffix is the buffer~$\buffer_{i}$. Player~$O$ has picked $\beta(0) \cdots \beta(\ell_{i-1}-1)$.

\end{itemize}

\begin{figure*}
\centering
\scalebox{1.2}{
\begin{tikzpicture}[thick, yscale = 1, shorten >=0pt]

\node at (0,2.2) {$\Gamma$};
\node at (.7,2.7) {$I\!\!:$};
\node at (.7,1.7) {$O\!\!:$};

\node at (0,-.5) {$\Gamma'$};
\node at (.7,0) {$I\!\!:$};
\node at (.7,-1) {$O\!\!:$};

\draw[black!30, dashed] (0,1.1) -- (13, 1.1);
\path[black!30, very thin] 
(6.5,-.95) edge[bend left = 9] (3.5,1.3)
(7.5,-.95) edge[bend left = 9] (4.5,1.3)
(2,2.1) edge[bend right = 9] (3, -.45);


\def\x{2}
\def\y{2.7}
\def\step{i-1}
\draw (\x,\y) -- (\x+3,\y);
\draw (\x,\y-.15) -- (\x,\y+.15);
\draw (\x+1,\y-.1) -- (\x+1,\y+.1);
\draw (\x+2,\y-.1) -- (\x+2,\y+.1);
\draw (\x+3,\y-.15) -- (\x+3,\y+.15);

\draw[thin] (\x-1, \y) -- (\x+.1,\y);
\draw[thin] (\x-1, \y-.15) -- (\x-1, \y+.15);

\node at (\x+.5,\y+.2) {\footnotesize $x_{\step}$};
\node at (\x+1.5,\y+.2) {\footnotesize $y_{\step}$};
\node at (\x+2.5,\y+.2) {\footnotesize $z_{\step}$};

\def\x{5}
\def\y{2.7}
\def\step{i}
\draw (\x,\y) -- (\x+3,\y);
\draw (\x,\y-.1) -- (\x,\y+.1);
\draw (\x+1,\y-.1) -- (\x+1,\y+.1);
\draw (\x+2,\y-.1) -- (\x+2,\y+.1);
\draw (\x+3,\y-.15) -- (\x+3,\y+.15);

\node at (\x+.5,\y+.2) {\footnotesize $x_{\step}$};
\node at (\x+1.5,\y+.2) {\footnotesize $y_{\step}$};
\node at (\x+2.5,\y+.2) {\footnotesize $z_{\step}$};


\begin{scope}[color = black]

\def\x{2}
\def\y{1.7}
\def\step{i-1}
\draw[dotted] (\x,\y) -- (\x+3,\y);
\draw (\x,\y-.15) -- (\x,\y+.15);
\draw (\x+1,\y-.1) -- (\x+1,\y+.1);
\draw (\x+2,\y-.1) -- (\x+2,\y+.1);
\draw (\x+3,\y-.15) -- (\x+3,\y+.15);

\draw[thin, black] (\x-1, \y) -- (\x-.025,\y);
\draw[thin] (\x-1, \y-.15) -- (\x-1, \y+.15);

\node at (\x+.5,\y-.15) {\footnotesize $x_{\step}'$};
\node at (\x+1.5,\y-.15) {\footnotesize $y_{\step}'$};
\node at (\x+2.5,\y-.15) {\footnotesize $z_{\step}'$};

\end{scope}

\draw[decorate, decoration={snake, segment length=.3cm, amplitude=.05cm}] (1,2.2) -- (2,2.2);
\draw[decorate, decoration={snake, segment length=.3cm, amplitude=.05cm}] (2,2.2) -- (3,2.2);
\node[draw=white,fill=white,inner sep =-.2] (start) at (1,2.2) {\scriptsize $q_{\initmark}$};
\node[draw=white,fill=white,inner sep =-.2] (start) at (2,2.2) {\scriptsize $q_{i-1}$};
\node[draw=white,fill=white,inner sep =-.2] (middle) at (3,2.2) {\scriptsize $q_{i-1}^*$};


\def\x{3}
\def\y{0}
\def\step{i-1}
\draw (\x,\y) -- (\x+2,\y);
\draw[dashed] (\x+2,\y) -- (\x+3,\y);
\draw (\x+3,\y) -- (\x+5,\y);

\draw (\x,\y-.15) -- (\x,\y+.15);
\draw (\x+1,\y-.1) -- (\x+1,\y+.1);
\draw (\x+2,\y-.1) -- (\x+2,\y+.1);
\draw (\x+3,\y-.1) -- (\x+3,\y+.1);
\draw (\x+4,\y-.1) -- (\x+4,\y+.1);
\draw (\x+5,\y-.15) -- (\x+5,\y+.15);

\draw[thin] (\x-2, \y) -- (\x,\y);
\draw[thin] (\x-2, \y-.15) -- (\x-2, \y+.15);

\node at (\x+.5,\y+.2) {\footnotesize $x_{\step}$};
\node at (\x+1.5,\y+.2) {\footnotesize $y_{\step}$};
\node at (\x+2.5,\y+.2) {\footnotesize $\cdots$};
\node at (\x+3.5,\y+.2) {\footnotesize $y_{\step}$};
\node at (\x+4.5,\y+.2) {\footnotesize $z_{\step}$};

\def\x{8}
\def\y{0}
\def\step{i}
\draw (\x,\y) -- (\x+2,\y);
\draw[dashed] (\x+2,\y) -- (\x+3,\y);
\draw (\x+3,\y) -- (\x+5,\y);

\draw (\x,\y-.15) -- (\x,\y+.15);
\draw (\x+1,\y-.1) -- (\x+1,\y+.1);
\draw (\x+2,\y-.1) -- (\x+2,\y+.1);
\draw (\x+3,\y-.1) -- (\x+3,\y+.1);
\draw (\x+4,\y-.1) -- (\x+4,\y+.1);
\draw (\x+5,\y-.15) -- (\x+5,\y+.15);

\node at (\x+.5,\y+.2) {\footnotesize $x_{\step}$};
\node at (\x+1.5,\y+.2) {\footnotesize $y_{\step}$};
\node at (\x+2.5,\y+.2) {\footnotesize $\cdots$};
\node at (\x+3.5,\y+.2) {\footnotesize $y_{\step}$};
\node at (\x+4.5,\y+.2) {\footnotesize $z_{\step}$};


\begin{scope}[color = black]

\def\x{3}
\def\y{-1}
\def\step{i-1}
\draw (\x,\y) -- (\x+2,\y);
\draw[dashed] (\x+2,\y) -- (\x+3,\y);
\draw (\x+3,\y) -- (\x+5,\y);

\draw (\x,\y-.15) -- (\x,\y+.15);
\draw (\x+1,\y-.1) -- (\x+1,\y+.1);
\draw (\x+2,\y-.1) -- (\x+2,\y+.1);
\draw (\x+3,\y-.1) -- (\x+3,\y+.1);
\draw (\x+4,\y-.1) -- (\x+4,\y+.1);
\draw (\x+5,\y-.15) -- (\x+5,\y+.15);

\draw[very thin, black] (\x-2, \y) -- (\x,\y);
\draw[thin] (\x-2, \y-.15) -- (\x-2, \y+.15);

\node at (\x+3.5,\y-.2) {\footnotesize $y_{\step}'$};
\node at (\x+4.5,\y-.2) {\footnotesize $z_{\step}'$};
\end{scope}

\draw[decorate, decoration={snake, segment length=.3cm, amplitude=.05cm}] (3,-.5) -- (6,-.5);

\node[draw=white,fill=white,inner sep =-.2] (startprime) at (3,-.5) {\scriptsize $q_{i-1}$};
\node[draw=white,fill=white,inner sep =-.2] (middleprime) at (6,-.5) {\scriptsize $q_{i-1}^*$};

\node at (1,.7) {\footnotesize $0$};
\node at (3,.7) {\footnotesize $\ell_{i-2}$};
\node at (8,.7) {\footnotesize $\ell_{i-1}$};
\node at (13,.7) {\footnotesize $\ell_{i}$};

\end{tikzpicture}	
}
\caption{\emph{The situation} (in solid lines): in $\Gamma$, Player~$I$ has picked $\overline{a_0}, \ldots,\overline{a_i}$, Player~$O$ has picked $\overline{b_0}, \ldots,\overline{b_{i-2}}$ (hidden in the thin part at the beginning), and $q_{i-1}$ is the state of $\aut$ reached when processing ${\overline{a_0} \choose \overline{b_0}} \cdots {\overline{a_{i-2}} \choose \overline{b_{i-2}}}$ (denoted by the curly line).  In $\Gamma'$, Player~$I$ has repeated $y_i$ sufficiently often so that Player~$O$ has provided an answer to $x_{i-1}y_{i-1}\cdots y_{i-1}z_{i-1}$, i.e., up to position~$\ell_{i-1}-1$, and $q_{i-1}^*$ is the state reached when processing $x_{i-1}$ and all but the last copy of $y_{i-1}$ (and the corresponding answers of Player~$O$) starting in $q_{i-1}$. \newline
By construction, there is an answer~$x_{i-1}'$ to $x_{i-1}$ such that processing ${x_{i-1} \choose x_{i-1}'}$ from $q_{i-1}$ brings $\aut$ to $q_{i-1}^*$ as well. The block~$b_{i-1}$ (dotted) is the concatenation of $x_{i-1}'$, and $y_{i-1}'z_{i-1}'$ from $\Gamma'$.}
\label{fig_equivalenceproof}
\end{figure*}
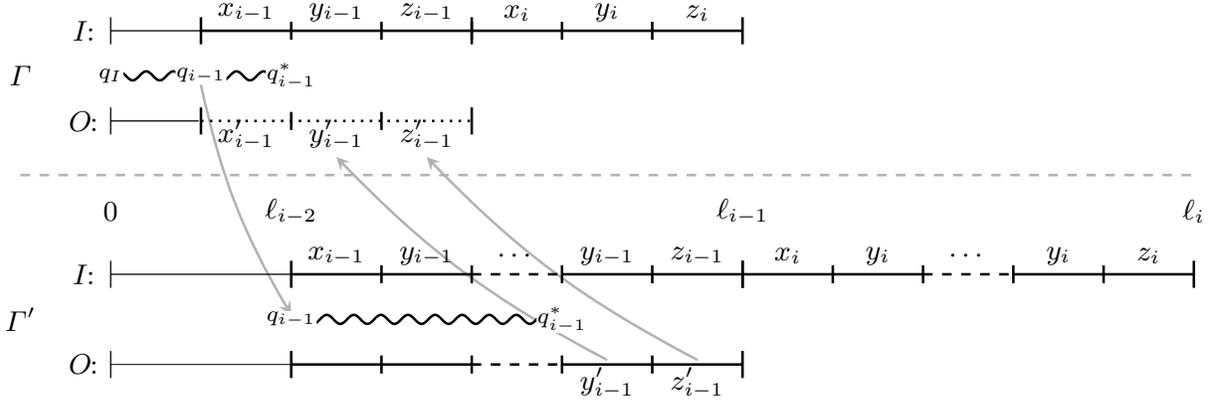

Now, let $i>0$ be arbitrary and let $q_{i-1}$ be the state reached by $\aut$ when processing ${ \overline{a_0} \choose \overline{b_0} } \cdots { \overline{a_{i-2}} \choose \overline{b_{i-2}} }$ from $q_\initmark$.  Furthermore, let $q_{i-1}^*$ be the state reached by $\aut$ when processing $x_{i-1}(y_{i-1})^{h_{i-1}-1}$ and the corresponding part of $\beta(0) \cdots \beta(\ell_{i-1}-1)$ starting in $q_{i-1}$, and let $t_{i-1}^*$ be the type of the run. Then, due to Remark~\ref{remark_powersetcharac} and the decomposition~$x_{i-1}y_{i-1}z_{i-1}$ being pumpable, there is a word~$x_{i-1}' \in \SigmaO^{\size{x_i-1}}$ such that $\aut$ reaches the same state~$q_{i-1}^*$ when processing ${x_{i-1} \choose x'_{i-1}}$ starting in $q_{i-1}$, and the run has type~$t_{i-1}^*$ as well. Now, we define $\overline{b_{i-1}} = x_{i-1}' y_{i-1}' z_{i-1}'$ where $y_{i-1}'$ and $z_{i-1}'$ are the letters picked by Player~$O$ at the positions of the last repetition of $y_{i_1}$ and at the positions of $z_{i_1}$, respectively (see Figure~\ref{fig_equivalenceproof}).

Using $\overline{b_{i-1}}$ we continue the simulation in $\Gamma$ by letting Player~$O$ pick the letters of $\overline{b_{i-1}}$ during the next $d$ rounds, which yields $d$ moves for Player~$I$ by applying $\stratI$. Call this sequence of letters~$\overline{a_{i+1}}$, which again has a pumpable decomposition~$x_{i+1} y_{i+1} z_{i+1}$. We again pick $h_{i+1} >0$ such that $\size{x_{i+1} (y_{i+1})^{h_{i+1}} z_{i+1}} \ge \sum_{j = \ell_{i-1}}^{\ell_{i}-1}f'(j)$ and define $\alpha(\ell_i) \cdots \alpha(\ell_{i+1}-1) = x_{i+1} (y_{i+1})^{h_{i+1}} z_{i+1}$. The strategy~$\stratI'$ for Player~$I$ in $\Gamma'$ is defined so that it picks the prefix of length~$\sum_{j=\ell_{i-1}}^{\ell_{i}-1}f'(j)$ of $\buffer_{i}\alpha(\ell_i ) \cdots \alpha(\ell_{i+1}-1)$ during the next $\ell_i$ rounds, independently of the choices of Player~$O$, and the remaining letters are stored in the buffer~$\buffer_{i+1}$. The prefix is again well-defined by the choice of $h_{i+1}$. Hence, during these rounds, Player~$O$ answers by producing $\beta(\ell_{i-1} ) \cdots \beta(\ell_i-1)$. Then, we are again in the situation described above for $i+1$.

We conclude by showing that $\stratI'$ is indeed a winning strategy for Player~$I$ in $\Gamma'$. To this end, let $w' = {\alpha \choose \beta}$ be an outcome of a play that is consistent with $\stratI'$ and let $w = { \overline{a_0} \choose \overline{b_0} }{ \overline{a_1} \choose \overline{b_1} }{ \overline{a_2} \choose \overline{b_2} } \cdots $ be the play in $\Gamma$ constructed during the simulation as described above, which is consistent with $\stratI$ and therefore winning for Player~$I$. Hence, the run of $\aut$ on $w'$ is rejecting.

A simple induction shows that $\aut$ reaches the same state when processing\[
{\alpha(0) \cdots \alpha(\ell_i -1) \choose \beta(0) \cdots \beta(\ell_i -1)  } \qquad\qquad \text{and} \qquad\qquad { \overline{a_0}\cdots \overline{a_i} \choose \overline{a_b}\cdots \overline{b_i}}, 
\]
call it $q_i$. By construction, the run of $\aut$ starting in $q_i$ processing ${\alpha(\ell_{i-1}) \cdots \alpha_(\ell_i-1) \choose \beta(\ell_{i-1}) \cdots \beta(\ell_i-1)}$ (using $\ell_{-1} = 0$) and the one starting in $q_i$ processing ${\overline{a_i} \choose \overline{b_i} }$ have the same type. The runs on the ${\overline{a_i} \choose \overline{b_i} }$  all have length~$d$, hence applying Lemma~\ref{lemma_runreplacement} to the rejecting run of $\aut$ on $w$ shows that the run of $\aut$ on $w'$ is rejecting as well. Thus, the outcome is winning for Player~$I$ and $\stratI$ is indeed a winning strategy for Player~$I$ in $\Gamma'$. 
 \end{proof}

Applying both directions of the equivalence proved in Theorem~\ref{thm_constantsuffices} and determinacy yields an upper bound on the necessary lookahead for Player~$O$.

\begin{corollary}
\label{cor_upperboundlookahead}
Let $L$ be recognized by a parity automaton with costs~$\aut$ with $n$ states and $k$ colors. If Player~$O$ wins $\delaygame{L}$ for some delay function~$f$, then also for the constant delay function~$f$ with $f(0) = 2^{2n^4k^2+1}$.
\end{corollary}

This upper bound can be slightly improved to $2^{2n^3k^2+1}$ by a more careful analysis: if $(q',(q_0,q_1,c_0, c_1, \ell)) \in \delta_P^+(\set{(q, \init_\aut(q))},w)$ for some $q \in Q$ and some $w \in \SigmaI^+$, then we have $q' = q_1$. Hence, not all states of $\aut'$ have to be considered when looking for a decomposition of some $\overline{a_i} $ into $ x_i y_i z_i$. 

Finally, the upper bound of Corollary~\ref{cor_upperboundlookahead} is asymptotically tight due to Proposition~\ref{prop_lowerbounds}, which is a generalization of the corresponding lower bound for delay games with winning conditions given by deterministic safety automata~\cite{KleinZimmermann16}.
\section{Determining the Winner}
\label{sec_solution}
The main result of this section is that the following problem is $\exptime$-complete: given a parity automaton with costs~$\aut$, does Player~$O$ win $\delaygame{L(\aut)}$ for some $f$? Hardness already holds for the special case of safety automata, thus we focus our attention on membership. To this end, we revisit the analogous result for delay games with classical parity conditions (i.e., without costs)~\cite{KleinZimmermann16}: such games are reduced to equivalent delay-free parity games of exponential size, which can be solved in exponential time (in the size of the original parity automaton). Here, we extend this proof to automata with costs while simplifying its structure. Furthermore, we obtain an exponential upper bound on the cost of a winning strategy for Player~$O$.

For the remainder of this section, fix $\aut = (Q, \SigmaI \times \SigmaO, q_\initmark, \delta, \col, \cost)$, let $\aut' = (Q', \SigmaI \times \SigmaO, q_\initmark', \delta', \col', \cost')$ and $\delta_P^+ \colon 2^{Q'} \times \SigmaI \rightarrow 2^{Q'}$ be defined as in Section~\ref{sec_constant}, and recall that $Q' = Q \times T_\aut$, where $T_\aut$ is the set of types of runs of $\aut$.

Given $x \in \SigmaI^+$, we define the function~$r_x \colon Q \rightarrow 2^{Q'}$ via
\[
r_x(q) = \delta_P^+(\set{(q, \init_\aut(q))},x).
\]
Now, we define $x \equiv_\aut x'$ if, and only if, $r_x = r_{x'}$, which is a finite equivalence relation. Furthermore, we can assign to every $\equiv_\aut$ equivalence class~$S$ a function~$r_S$ from $Q$ to $2^{Q'}$, i.e., $r_S = r_x$ for all $x \in S$, which is independent of representatives. Finally, let $\R$ denote the set of $\equiv_\aut$ equivalence classes of words in $\SigmaI^{2d}$, where $d = 2^{2n^4k^2}$ as before.

Next, we construct a delay-free game~$\game(\aut)$ between Player~$I$ and Player~$O$ that is won by Player~$O$ if, and only if, she wins $\delaygame{L(\aut)}$ for some delay function~$f$. The game~$\game(\aut)$ is a zero-sum infinite-duration two-player game of perfect information played in rounds~$i = 0,1,2, \ldots$. Intuitively, Player~$I$ picks a sequence~$S_0 S_1 S_2 \cdots$ of equivalence classes from $\R$, which induces an infinite word~$\alpha$ over $\SigmaI$ by picking representatives. Player~$O$ implicitly picks an infinite word over $\SigmaO$ by constructing a run of $\aut$ on a word over $\SigmaI \times \SigmaO$ whose projection to $\SigmaI$ is $\alpha$. She wins, if the run is accepting. To account for the delay, she is always one move behind. 

Formally, in round~$0$, Player~$I$ picks an equivalence class~$S_0 \in \R$ and then Player~$O$ has to pick $(q_0, t_0) = q_\initmark'$, the initial state of $\aut'$.\footnote{This move is trivial, but we add it to keep the definition consistent.} In round~$i>0$, Player~$I$ picks an equivalence class $S_i \in \R$ and then Player~$O$ picks a state~$(q_i,t_i) \in r_{S_{i-1}}(q_{i-1})$ (due to completeness of $\aut$, Player~$O$ always has an available move).

For every $t \in T_\aut$ such that there is some run of $\aut$ of type~$t$, fix one such run~$\rho_t$. Now, consider a play~$\pi = S_0 (q_0, t_0) S_1 (q_1, t_1) S_2 (q_2, t_2) \cdots $ of $\game(\aut)$. By construction, $\rho_{t_{i+1}}$ starts in $q_i$ and ends in $q_{i+1}$, for every $i \ge 0$. The play~$\pi$ is winning for Player~$O$ if, and only if, the run $\rho_{t_1}\rho_{t_2}\rho_{t_3} \cdots$ of $\aut$ is accepting (note that $t_0$ is disregarded). As the length of the representatives~$\rho_{t_i}$ is bounded (there are only finitely many), Corollary~\ref{corollary_runreplacement} implies that the winner is independent of the choice of the representatives.  

A strategy for Player~$I$ in $\game(\aut)$ is a mapping~$\stratI \colon (\R \cdot Q')^* \rightarrow \R$ while a strategy for Player~$O$ is a mapping~$\stratO \colon (\R \cdot Q')^* \cdot \R \rightarrow Q'$ that has to satisfy $\stratO(S_0) = q_\initmark'$ for every $S_0 \in \R$ and $\stratO(S_0 \cdots (q_i, t_i)S_{i+1}) \in r_{S_i}(q_i)$ for all $S_0 \cdots (q_i, t_i)S_{i+1} \in (\R \cdot Q')^+ \cdot \R$. A play $S_0 (q_0, t_0) S_1 (q_1, t_1) S_2 (q_2, t_2) \cdots$ is consistent with $\stratI$ if $S_i = \stratI(S_0 \cdots (q_{i-1}, t_{i-1}))$ for every $i$, and it is consistent with $\stratO$, if $(q_i, t_i) = \stratO(S_0 \cdots S_i)$ for every $i$. A strategy is winning for a Player~$P \in \set{I,O}$, if every play that is consistent with the strategy is won by Player~$P$. As usual, we say that Player~$P$ wins $\game(\aut)$, if she has a winning strategy. 

\begin{lemma}
\label{lemma_equivalence}
Player~$O$ wins $\game(\aut)$ if, and only if, she wins $\delaygame{L(\aut)}$ for some $f$.
\end{lemma}

\begin{proof}
For the sake of simplicity, we denote $\game(\aut)$ by $\game$ and $\delaygame{L(\aut)}$ by $\Gamma$, provided $f$ is clear from context.

First, let Player~$O$ win $\delaygame{L(\aut)}$ for some $f$. Then, due to Theorem~\ref{thm_constantsuffices}, she also wins $\Gamma = \delaygame{L(\aut)}$ for the constant delay function~$f$ with $f(0) = 2d$ (recall $d = 2^{2n^4k^2}$). Thus, fix this $f$ and let $\stratO$ be a winning strategy for Player~$O$ in $\Gamma = \delaygame{L(\aut)}$. We construct a winning strategy~$\stratO'$ for Player~$O$ in $\game$ by simulating a play in $\game$ by a play in $\Gamma$. 

Thus, let $S_0 \in \R$ be a first move of Player~$I$ in $\game$. In round~$0$, Player~$O$ has to pick $(q_0, t_0) = q_\initmark'$. Hence, we define $\stratO'(S_0) = q_\initmark'$, independently of the pick $S_0$ by Player~$I$. Now, let $S_1 \in \R$ be the second move of Player~$I$ in reaction to Player~$O$ picking $(q_0, t_0)$. By $S_0, S_1 \in \R$, there are words~$\alpha(0) \cdots \alpha(f(0)-1) \in S_0$ and $\alpha(f(0)) \cdots \alpha(2f(0)-1) \in S_1$, both of length~$f(0)$. We simulate the play prefix~$S_0 q_0 S_1$ of $\game$ in $\Gamma$ by letting Player~$I$ pick $\alpha(0) \cdots \alpha(f(0)-1)$ in round~$0$ of $\Gamma$ and the letters of $\alpha(f(0)) \cdots \alpha(2f(0)-1)$ during the next $f(0)$ rounds. Applying the winning strategy~$\stratO$ for Player~$O$ in $\Gamma$ to these moves yields $f(0) +1$ letters~$\beta(0) \cdots \beta(f(0))$.

Then, we are in the following situation for $i =1$:

\begin{itemize}
	\item In $\game$, we have a play prefix~$S_0 (q_0, t_0) S_1 \cdots (q_{i-1}, t_{i-1}) S_1$, and
	\item in $\Gamma$, Player~$I$ has picked $\alpha(0) \cdots \alpha((i+1)\cdot f(0)-1)$ during the first $1 + i\cdot f(0)$ rounds while Player~$O$ has picked $\beta(0) \cdots \beta(i\cdot f(0))$. 
\end{itemize}

Now, let $i>0$ be arbitrary and let $q_{i}$ be the state reached by $\aut$ when processing 
\[
{ \alpha((i-1)\cdot f(0)) \cdots \alpha(i\cdot f(0)-1) \choose \beta((i-1)\cdot f(0)) \cdots \beta(i \cdot f(0)-1) }
\]
from $q_{i-1}$, and let $t_i$ be the type of the corresponding run. Then, we have $(q_i,t_i) \in r_{S_{i}}(q_{i-1})$ by construction and define $\stratO(S_0 (q_0, t_0) S_1 \cdots (q_{i-1}, t_{i-1}) S_i) = (q_i, t_i)$. This move is again answered by Player~$I$ in $\game$ by picking $S_{i+1} \in \R$, which induces $\alpha((i+1)\cdot f(0)) \cdots \alpha((i+2)\cdot f(0)-1) \in S_{i+1}$. We continue the simulation by letting Player~$I$ pick the letters of $\alpha((i+1)\cdot f(0)) \cdots \alpha((i+2)\cdot f(0)-1)$ during the next $f(0)$ rounds, which is again answered by letters~$\beta(i\cdot f(0)) \cdots \alpha((i+1)\cdot f(0)-1)$ according to $\stratO$. Thus, we are in the same situation as above for $i+1$, which concludes the definition of $\stratO'$.

It remains to show that $\stratO'$ is a winning strategy for Player~$O$ in $\game$: let $S_0 (q_0,t_0) S_1 (q_1, t_1) S_2 (q_2, t_2) \cdots $ be a play that is consistent with $\stratO'$ and let ${\alpha \choose \beta}$ be the outcome of the corresponding play in $\Gamma$ constructed during the simulation, which is consistent with $\stratO$. By construction, $t_{i+1}$ for $i \ge 0$ is the type of the run of $\aut$ on
\[{ \alpha((i-1)\cdot f(0)) \cdots \alpha(i\cdot f(0)-1) \choose \beta((i-1)\cdot f(0)) \cdots \beta(i\cdot f(0)-1) }\]
starting in $q_{i}$, which ends in $q_{i+1}$. We call this finite run~$\rho_i$. Hence, $\rho_0 \rho_1 \rho_2 \cdots$ is the run of $\aut$ on ${\alpha \choose \beta}$, which is accepting due to ${ \alpha \choose \beta }$ being the outcome of a play that is consistent with the winning strategy~$\stratO$. 

As each $\rho_i$ has length~$f(0)$, Corollary~\ref{corollary_runreplacement} is applicable to the runs~$\rho_0 \rho_1 \rho_2 \cdots$ and $\rho_{t_1} \rho_{t_2} \rho_{t_3} \cdots$: as the former run is accepting, the latter is as well. Hence, the play~$S_0 (q_0,t_0) S_1 (q_1, t_1) S_2 (q_2, t_2) \cdots $ of $\game$ is winning for Player~$O$. Thus, $\stratO$ is indeed a winning strategy for her in $\game$. 

For the other direction, let $\stratO'$ be a winning strategy for Player~$O$ in $\game$. We construct a winning strategy~$\stratO$ for Player~$O$ in $\Gamma = \delaygame{L(\aut)}$ for the unique constant delay function~$f$ with $f(0) = 4d$ by simulating a play of $\Gamma$ by a play of $\game$. 

Thus, let $\alpha(0) \cdots \alpha(f(0)-1)$ be the move of Player~$I$ in round~$0$ of $\Gamma$ and let $S_0 \in \R$ be the equivalence class of $\alpha(0) \cdots \alpha(2d-1)$ as well as $S_1$ the equivalence class of $\alpha(2d) \cdots \alpha(4d-1)$. Now, consider the following play prefix in $\game$: Player~$I$ picks $S_0$, then Player~$O$ picks $(q_0, t_0) = q_\initmark'$ according to $\stratO'$, then Player~$I$ picks $S_1$, and finally Player~$O$ picks $(q_1, t_1) = \stratO(S_0 (q_0, t_0) S_1)$ according to $\stratO'$. 

Then, we are in the following situation for $i = 1$:
\begin{itemize}
	\item In $\Gamma$, Player~$I$ has picked $\alpha(0) \cdots \alpha((i+1)\cdot (2d)-1)$ and Player~$O$ has picked $\beta(0) \cdots \beta((i-1)\cdot (2d)-1)$.
	\item In $\game$, we have a play prefix~$S_0 (q_0, t_0) S_1 \cdots (q_{i-1}, t_{i-1}) S_i (q_i, t_i)$.
\end{itemize}

Now, let $i>0$ be arbitrary. Due to $(q_i, t_i) \in r_{S_i}(q_{i-1})$ and $\alpha((i-1)\cdot (2d)) \cdots \alpha(i\cdot (2d)-1) \in S_{i-1}$, there is a word~$\beta((i-1)\cdot (2d)) \cdots \beta(i\cdot (2d)-1)$ such that the run of $\aut$ on 
\[
{\alpha((i-1)\cdot (2d)) \cdots \alpha(i\cdot (2d)-1) \choose \beta((i-1)\cdot (2d))\cdots  \beta(i\cdot (2d)-1)}
\]
starting in $q_{i-1}$ ends in $q_i$ and has type~$t_i$. We define $\stratO$ so that it picks the letters of $\beta((i-1)\cdot (2d))\cdots  \beta(i\cdot (2d)-1)$ during the $2d$ rounds $(i-1)\cdot (2d), \ldots, i\cdot (2d)-1$. During these rounds, Player~$I$ picks $\alpha((i+1)\cdot (2d)) \cdots \alpha((i+2)\cdot (2d)-1)$. Let $S_{i+1}$ be its equivalence class. 

Then, 
we continue the play prefix~$S_0 (q_0, t_0) S_1 \cdots (q_{i-1}, t_{i-1}) S_i (q_i, t_i)$ in $\game$ by letting Player~$I$ pick $S_{i+1}$ and by letting Player~$O$ pick $(q_{i+1}, t_{i+1}) = \stratO'(S_0 (q_0, t_0) S_1 \cdots (q_{i-1}, t_{i-1}) S_i (q_i, t_i) S_{i+1})$. Then, we are in the same situation as above for $i+1$, which concludes the definition of $\stratO$. 

Again, it remains to prove that $\stratO$ is indeed winning for Player~$O$ in $\Gamma$. To this end, let ${\alpha \choose \beta}$ be the outcome of a play that is consistent with $\stratO$ and let $S_0 (q_0,t_0) S_1 (q_1,t_1) S_2 (q_2,t_2) \cdots $ be the corresponding play constructed during the simulation, which is consistent with $\stratO'$. By construction, the run of $\aut$ on
\[
{\alpha((i-1)\cdot (2d)) \cdots \alpha(i\cdot (2d)-1) \choose \beta((i-1)\cdot (2d))\cdots  \beta(i\cdot (2d)-1)}
\]
starting in $q_{i-1}$ has type $t_i$ and ends in $q_i$ (for $i>0$). Call this run $\rho_i$, i.e., the run of $\aut$ on ${\alpha \choose \beta}$ is $\rho_1 \rho_2 \rho_3 \cdots$. 
Thus, Corollary~\ref{corollary_runreplacement} is applicable to the runs $\rho_1 \rho_2 \rho_3 \cdots$ and $\rho_{t_1} \rho_{t_2} \rho_{t_3} \cdots$: as the latter one is accepting due to $S_0 (q_0,t_0) S_1 (q_1,t_1) S_2 (q_2,t_2) \cdots $ being consistent with a winning strategy for Player~$O$ in $\game$, the former one is accepting as well. Hence, the outcome~${\alpha \choose \beta}$ is accepted by $\aut$, which implies that the corresponding play is winning for Player~$O$. Therefore, $\stratO$ is indeed a winning strategy for Player~$O$ in $\Gamma$. 
 \end{proof}

Now, we are able to state and prove our main theorem of this section. 

\begin{theorem}
\label{thm_decidability}
The following problem is $\exptime$-complete: given a parity automaton with costs~$\aut$, does Player~$O$ win $\delaygame{L(\aut)}$ for some $f$?
\end{theorem}

\begin{proof}
We focus on membership as $\exptime$-hardness already holds for safety automata~\cite{KleinZimmermann16}.
To this end, we show how to model the abstract game~$\game(\aut)$ as an arena-based parity game with costs~\cite{FijalkowZimmermann14} of exponential size with at most one more color than $\aut$. This game can be constructed (argued below) and solved in exponential time (in the size of $\aut$)~\cite{FijalkowZimmermann14}. Lemma~\ref{lemma_equivalence} shows that solving this game yields the correct answer. Hence, the problem is in $\exptime$.

Intuitively, the arena encodes the rules of $\game(\aut)$: the players pick equivalence classes from $\R$ and states from $Q'$ in alternation. The restrictions on the states that may be picked are enforced by storing the last equivalence class and the last state in the vertices of the arena. Finally, to encode the winning condition of $\game(\aut)$, we simulate the effect of a run of type~$t = (q,q', c, c', \ell)$ every time Player~$O$ picks a state~$(q,t)$. Recall that $c$ encodes the largest answer, $c'$ the largest unanswered request, and $\ell$ whether the overall cost is zero or greater than zero. The effect is simulated by first visiting a state of color~$c$, then one of color~$c'$, and equipping the edge between these vertices with cost~$\ell$. Afterwards, Player~$I$ again picks another equivalence class. 

Formally, we define the parity game with costs~$(\arena, \cp(\col))$ with arena~$\arena = (V, V_I, V_O, E, \cost)$ where
\begin{itemize}
	\item $V = V_I \cup V_O$ with $V_I = \set{v_\initmark} \cup \R \times Q' \times\set{0,1}$ and $V_O = \R \times Q' \times \R$, and
	\item $E$ is the union of the following sets of edges:
		\begin{itemize}
			\item $\set{(v_\initmark, (S_0, q_\initmark', S_1)) \mid S_0, S_1 \in \R}$: the initial moves of Player~$I$ (which subsume the first (trivial) move of Player~$O$),
			\item $\set{((S_0,(q_0,t_0),0),(S_0,(q_0,t_0),1)) \mid S_0 \in \R, (q_0,t_0) \in Q'}$: deterministic moves of Player~$I$ used to simulate the effect of a run of type~$t_0$,
			\item $\set{((S_0,(q_0,t_0),1),(S_0,(q_0,t_0),S_1)) \mid S_0, S_1 \in \R, (q_0,t_0) \in Q'}$: regular moves of Player~$I$ picking the next equivalence class~$S_1$, and
			\item $\set{((S_0, (q_0, t_0), S_1),(S_1, (q_1, t_1),0)) \mid S_0, S_1, \in \R, (q_0, t_0), (q_1, t_1) \in Q', (q_1, t_1) \in r_{S_0}(q_0)}$: moves of Player~$O$ picking the next state~$(q_1, t_1)$.
		\end{itemize}
	\item We define $\cost(e) = \inc $, if, and only if, $e = ((S_0,(q_0,t_0),0),(S_0,(q_0,t_0),1))$ with $t_0 = (q,q',c, c', \inc)$ for some $q,q' \in Q$ and $c,c' \in \col(Q)$, i.e., we simulate the cost encoded in the type~$t_0$.
	\item Finally, for $t = (q,q', c_0, c_1, \ell)$ we define $\col(S, (q, t),0) = c_0$ and $\col(S, (q, t),1) = c_1$, i.e., we simulate the largest response and afterwards the largest unanswered request encoded in $t$. Every other state has color~$0$, which has no effect on the satisfaction of the parity condition with costs, as it is too small to answer requests.
\end{itemize}

As an illustration of the construction, consider Figure~\ref{fig_arenaconstruction}, which depicts the vertices reached while simulating a play prefix of $\game(\aut)$. Note that the infix $(S_1, (q_1, t_1), 0 )(S_1, (q_1, t_1), 1 )$ has largest response~$c_0^1$ and largest unanswered request~$c_1^1$, just as encoded by the type~$t_1$. Similarly, the cost of this infix is the one encoded by $t_1$. All other vertices and edges are neutral. Furthermore, the type~$t_0$ encoded in the initial state~$q_\initmark' = (q_\initmark, \init_\aut(q_\initmark))= (q_0, t_0) $ is not simulated, just as it is ignored when it comes to determining the winner of a play in $\game(\aut)$.

\begin{figure*}
\centering	
\begin{tikzpicture}[thick]
\def\y{-1.8}
\node[p1,label=north:$0$] at (0,0) (a) {$v_\initmark$};
\node[p0,ellipse, inner sep = 0pt,label=north:$0$] at (3,0)(b) {$(S_0, (q_0, t_0), S_1)$};
\node[p1,label=north:$c_0^1$] at (7,0) (c) {$(S_1, (q_1, t_1),0)$};
\node[p1,label=north:$c_1^1$] at (11,0) (d) {$(S_1, (q_1, t_1),1)$};
\node[p0,ellipse, inner sep = 0pt,label=north:$0$] at (3,\y) (e) {$(S_1, (q_1, t_1), S_2)$};
\node[p1,label=north:$c_0^2$] at (7,\y) (f) {$(S_2, (q_2, t_2),0)$};
\node[p1,label=north:$c_1^2$] at (11,\y) (g) {$(S_2, (q_2, t_2),1)$};
\node at (14,\y) (h) {$\cdots$};

\path
(a) edge node[above] {$\eps$} (b)
(b) edge node[above] {$\eps$}(c)
(c) edge node[above] {$\ell_1$}(d)
(e) edge node[above] {$\eps$}(f)
(f) edge node[above] {$\ell_2$}(g)
(g) edge (h);	
\path[draw, rounded corners,->] (d.east) -| (12.5,-.8) -- node[above]{$\eps$} (1,-.8) |- (e.west);
\end{tikzpicture}
\caption{The construction of $\arena$: a play prefix~$S_0, (q_0, t_0) S_1, (q_1, t_1), S_2, (q_2, t_2)$ of~$\game(\aut)$  with $t_i = (q_i, q_i', c_0^i, c_1^i, \ell_i)$ is simulated by the depicted sequence of vertices. Colors are depicted above vertices, edge weights above edges. }
\label{fig_arenaconstruction}
\end{figure*}
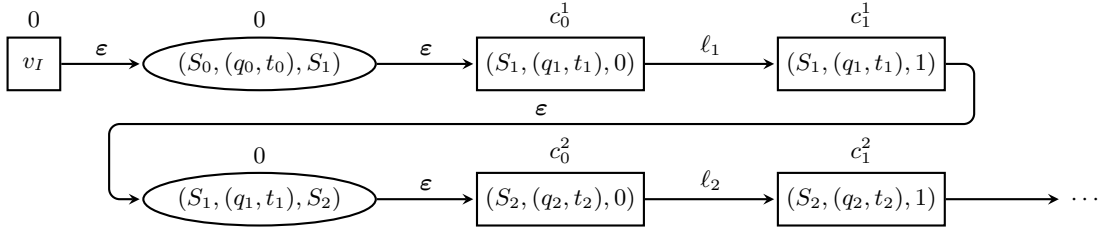

Thus, Corollary~\ref{corollary_runreplacement}, implies that Player~$O$ wins $\game(\aut)$ if, and only if, Player~$O$ has a winning strategy for the parity game with costs~$(\arena, \cp(\col))$ from $v_\initmark$. 
The construction of $(\arena, \cp(\col))$ is possible in exponential time using the same automata construction to determine the elements of $\R$ as in the case of plain parity conditions~\cite{KleinZimmermann16}.
 \end{proof}

If Player~$O$ wins an arena-based parity game with costs, then there is also a winning strategy for her whose cost is bounded by the number of vertices of the arena~\cite{FijalkowZimmermann14,WeinertZimmermann16}. Hence, an application of Lemma~\ref{lemma_runreplacement} yields an exponential upper bound on the cost of a winning strategy in a delay game with such a winning condition.

\begin{corollary}
\label{cor_upperboundlookaheadandbound}
Let $\aut$ be a parity automaton with costs with $n$ states and $k$ colors. If Player~$O$ wins $\delaygame{L(\aut)}$ for some $f$, then she also wins $\delaygame{L(\aut)}$ for the constant delay function~$f$ given by $f(0) = 2^{2n^4k^2+2}$ with a winning strategy~$\stratO$ satisfying $\cost_\aut(\stratO) \le n^3k^22^{2n^7k^4+3}$.
\end{corollary}

Due to Proposition~\ref{prop_lowerbounds}, these bounds are asymptotically tight.

\section{Streett Conditions}
\label{sec_streett}
In this section, we consider the more general case of delay games with winning conditions given by automata with finitary Streett acceptance or with Streett conditions with costs. In a parity condition, the requests and responses are hierarchically ordered.  Streett conditions generalize parity conditions by giving up this hierarchy. 

Formally, a Streett automaton with costs is a tuple~$\aut = (Q, \Sigma, q_\initmark, \delta, (Q_j, P_j)_{j\in J}, (\cost_j)_{j\in J}) $ where $Q$, $\Sigma$, $q_\initmark$, and $\delta$ are defined as for parity automata with costs. Furthermore, the acceptance condition~$(Q_j, P_j)_{j\in J}$ consists of a finite collection of Streett pairs~$(Q_j, P_j)$ of subsets~$Q_j, P_j \subseteq Q$. Here, states in $Q_j$ are requests of condition~$j$ which are answered by visiting a response in~$P_j$. Finally, $ (\cost_j)_{j\in J}$ is a collection of cost functions for $\aut$, one for each Streett pair. The size of $\aut$ is defined as $\size{Q} + \size{J}$.

For a run~$(q_0, a_0, q_1)(q_1, a_1, q_2)(q_2, a_2, q_3)\cdots$ and a position~$n$, we define the cost-of-response~$\streettdist_j(\rho, n)$ of pair~$j$ to be $0$, if $q_n \notin Q_j$, and to be
\begin{align*}
\min\set{ \cost_j((q_n,a_n,q_{n+1})\cdots(q_{n'-1},a_{n'-1}, q_{n'})) \mid n' \ge n \text{ and } q_{n'} \in P_j},
\end{align*}
if $q_n \in Q_j$. Furthermore, we aggregate these costs by defining 
\[
\streettdist(\rho, n)  = \max\nolimits_{j\in J} \streettdist_j(\rho, n).
\]
Finally, we say that $\rho$ is accepting if it satisfies the Streett condition with costs, i.e., if $\limsup_{n \to\infty}\streettdist(\rho,n) < \infty$. We recover classical Streett acceptance as the special case where every edge is an $\epsilon$-edge w.r.t.\ every cost function. Similarly, finitary Streett~\cite{ChatterjeeHenzingerHorn09} acceptance is the special case where every edge is an increment-edge w.r.t.\ every cost function. 

In the following, we consider delay games with winning conditions specified by such automata. The notion of the cost of a strategy for Player~$O$ and that of optimality is defined as in the case of parity automata. Also, note that every parity condition is a Streett condition, i.e., all lower bounds already proven hold for Streett conditions as well. In particular, exponential lookahead is necessary to win delay games with finitary Streett conditions and solving such games is $\exptime$-hard. We complement these lower bounds by doubly-exponential upper bounds, both on the necessary lookahead and on the solution complexity. It is open whether this gap can be closed. Nevertheless, we show in Section~\ref{sec_tradeoffs} that doubly-exponential lookahead is necessary for \emph{optimal} strategies!

Our first step towards these results is the generalization of the replacement lemma for parity conditions with costs. To this end, we have to generalize the notion of types. To this end, fix a Streett automaton with costs~$\aut = (Q, \Sigma, q_\initmark, \delta, (Q_j, P_j)_{j\in J}, (\cost_j)_{j\in J})$. The set of types of $\aut$ is defined as $T_\aut = Q^2 \times \set{\bot, p, q, pq}^J \times \set{\eps,\inc}^J$. The type of a non-empty finite run~$(q_{0}, a_{0}, q_{1}) \cdots (q_{n_1}, a_{n-1}, q_{n})$ is defined as $(q_0, q_n, g, \ell )$ where 
\begin{itemize}
	\item $g(j) = pq$ if the run contains a response of condition~$j$ as well as an unanswered request of condition~$j$,
	\item $g(j) = p$ if the run contains a response of condition~$j$, but no unanswered request of condition~$j$,
	\item $g(j) = q$ if the run contains no response of condition~$j$, but an unanswered request of condition~$j$, and
	\item $g(j) = \bot$ if the run contains neither a request nor a response of condition~$j$.
\end{itemize}
Furthermore, $\ell(j)$ is equal to $\inc$ if, and only if, the run contains an increment-transition with respect to $\cost_j$.

With this definition, the replacement property formalized in Lemma~\ref{lemma_runreplacement} also holds for runs of Streett automata with costs, which is proven using essentially the same argument as for parity automata with costs. Similarly, Corollary~\ref{corollary_runreplacement} holds for Streett automata with costs as well. Also, as for parity automata with costs, the type of a run can be computed on the fly using functions~$\init_\aut$ and $\update_\aut$ with the same properties as their analogues in Remark~\ref{remark_tracking}.

Using these results, we determine upper bounds on the necessary lookahead and the complexity of solving delay games induced by Streett automata with costs. Here, the exponential increase in complexity in comparison to games induced by parity automata with costs stems from the fact that there are exponentially many types for Street automata, but only polynomially many for parity~automata. 

\begin{theorem}
\label{thm_constantsufficesstreett}
Let $L$ be recognized by a Streett automaton with costs~$\aut$ with $n$ states and $k$ Streett pairs, and let $f$ be the constant delay function with $f(0) = 2^{n^42^{3k}+1}$. The following are equivalent:
\begin{enumerate}

	\item Player~$I$ wins $\delaygame{L(\aut)}$.

	\item Player~$I$ wins $\delaygamep{L(\aut)}$ for every delay function~$f'$.

\end{enumerate}
\end{theorem}

\begin{proof}
Similar to the one of Theorem~\ref{thm_constantsuffices} using $d = \size{({2^{Q \times T_\aut}})^{Q}} = 2^{n^42^{3k}}$.
\end{proof}

Thus, we obtain $f(0) = 2^{n^42^{3k}+1}$ as an upper bound on the necessary constant lookahead for Player~$O$ to win a delay game with winning condition~$L$. 

Also, the decidability proof for parity conditions is applicable to Streett conditions, again with an exponential blowup. 

\begin{theorem}
\label{thm_decidabilitystreett}
The following problem is in $\twoexp$: given a Streett automaton with costs~$\aut$, does Player~$O$ win $\delaygame{L(\aut)}$ for some $f$?
\end{theorem}

\begin{proof}
Given $\aut$, one constructs an abstract game~$\game(\aut)$ as for the parity case and proves the analogue of Lemma~\ref{lemma_equivalence}. Then, one models $\game(\aut)$ as an arena-based Streett game with costs of doubly-exponential size with the same number of Streett pairs as $\aut$, which can be solved in doubly-exponential time~\cite{FijalkowZimmermann14}.
\end{proof}

Again, modeling the abstract game as an arena-based game yields an upper bound on the cost of a winning strategy for a delay game with winning condition given by a Streett automaton with costs: in an arena-based Streett game with costs, there is a tight exponential bound on the cost of an optimal strategy~\cite{FijalkowZimmermann14,WeinertZimmermann17}. This implies a triply-exponential upper bound for the original delay game.

\section{Trading Lookahead for Costs}
\label{sec_tradeoffs}
Introducing lookahead allows Player~$O$ to win games she loses in a delay-free setting. In this section, we study another positive effect of lookahead: it allows to reduce the cost of optimal strategies, i.e., one can trade lookahead for quality and vice versa. To simplify our notation, let $f_k$ for $k \ge 0$ denote the unique constant delay function with $f_k(0) = k+1$. Thus, $k$ denotes the size of the lookahead. In particular, a delay game~$\delaygamec{0}{L}$ is a delay-fee game.

\subsection{Tradeoffs for Parity Conditions}
\label{subsec_tradeoffs_parity}

First, we consider parity conditions and show that already the smallest possible lookahead allows to improve the cost of an optimal strategy from~$\size{\aut}$ to $1$.

\begin{theorem}
\label{thm_tradeoff}
For every $n>0$, there is a language~$L_n$ recognized by a finitary Büchi automaton with costs~$\aut_n$ with $n+2$ states such that
\begin{itemize}
	\item an optimal strategy for $\delaygamec{0}{L_n}$ has cost~$n$, but
	\item an optimal strategy for $\delaygamec{1}{L_n}$ has cost~$1$. 
\end{itemize}
\end{theorem}

\begin{proof}
Consider the finitary Büchi automaton~$\aut_n$ depicted in Figure~\ref{fig_tradeoff} over $\SigmaI \times \SigmaO = \set{0,1}^2$. Every run of $\aut_n$ visits the initial state infinitely often, which answers all requests. Thus, consider a run starting and ending in the initial state, but not visiting it in between. There are two types of such runs, those of length two and those of length~$n+1$ (visiting the gray state). Runs of the first type process a word of the form~${* \choose b}{b \choose *}$ for some $b \in \set{0,1}$ and an arbitrary letter~$*$, runs of the second type a word of the form~${* \choose b}{1-b \choose *}{* \choose *}^{n-1}$. A run having infinitely many infixes of the second type has cost~$n$, otherwise it has cost~$1$. Thus, to achieve cost~$1$, Player~$O$ has to predict the next move of Player~$I$. This is possible with constant lookahead~$1$, but not without lookahead. 
 \end{proof}

\begin{figure*}
\begin{floatrow}
\floatbox{figure}[.3\textwidth][\FBheight][t]
{\caption{The automaton~$\aut_n$ for the proof of Theorem~\ref{thm_tradeoff}. The path from the gray state to the doubly-lined state has $n-1$ edges and $*$ denotes an arbitrary letter.}
\label{fig_tradeoff}}
{\scalebox{1.2}{
\begin{tikzpicture}[thick]
\tikzset{state/.style = {shape = circle, draw, thick, minimum size = 0.6cm}}

\node[state]	 				at (1.5,0) (n) {};
\node[state] 					at (-1.5,0) (s) {};
\node[state, accepting] 		at (0,0) (c) {};

\node[state]					at (0,-1.5) (l1) {};
\node[state, fill=black!20]					at (0,-4.5) (ln) {};
\node							at (0,-3) (ldots) {$\vdots$};

\path[->] (0,0.6) edge (c);
\path[->]
(ln) edge node[right] {${* \choose *}$} (ldots)
(ldots) edge node[right] {${* \choose *}$} (l1)
(l1) edge node[right,near start] {${* \choose *}$} (c)
(c) edge[bend left] node[above,near start] {${* \choose 0}$} (n)
(c) edge[bend right] node[above,near start] {${* \choose 1}$} (s)
(n) edge[bend left] node[below,near start] {${0 \choose *}$} (c)
(s) edge[bend right] node[below,near start] {${1 \choose *}$} (c);

\path[draw, ->, rounded corners] (n.south) -- node[right, near end] {${1 \choose *}$} (.7,-4.5) -- (ln.east);
\path[draw, ->, rounded corners] (s.south) -- node[left, near end] {${0 \choose *}$} (-.7,-4.5) -- (ln.west);

\end{tikzpicture}	
}}\hspace*{.02\textwidth}
\floatbox{figure}[.65\textwidth][\FBheight][t]
{         \caption{The automaton~$\aut_3$ for the proof of Theorem~\ref{thm_tradeoff_gradual}. A transition labeled by ${ * \choose *}^2$ represents a path of two transitions, each labeled with ${ * \choose *}$, where $*$ denotes an arbitrary letter. The missing transitions of the gray states lead to a sink state of color~$1$.}
          \label{fig_tradeoff_gradual}}
{\scalebox{1.2}{
\begin{tikzpicture}[thick]
\tikzset{state/.style = {shape = circle, draw, thick, minimum size = 0.6cm}}
\def\x{1.4}
\def\y{1}
\def\yoffset{.2}
\node[state, accepting] 	(c) 	at (2*\x,0) {};

\node[state]				(30)	at (1*\x,\y+\yoffset) {};
\node[state]				(20)	at (2*\x,2*\y) {};

\node[state]				(31)	at (1*\x,-\y-\yoffset) {};
\node[state]				(21)	at (2*\x,-2*\y) {};
%
%
\node[state, fill = black!20]				(30p)	at (0,\y+\yoffset) {};
\node[state, fill = black!20]				(20p)	at (-1*\x,2*\y) {};
\node[state]				(10p)	at (-2*\x,3*\y-\yoffset) {};
\node[state, fill = black!20]				(31p)	at (0,-\y-\yoffset) {};
\node[state, fill = black!20]				(21p)	at (-1*\x,-2*\y) {};
\node[state]				(11p)	at (-2*\x,-3*\y+\yoffset) {};
\node[state]				(l1)	at (0,0) {};
\node[state]				(l2)	at (-1*\x,0) {};
\node[state]				(l3)	at (-2*\x,0) {};
\node[state]				(l4)	at (-3*\x,0) {};

\path[->] (2*\x+.6,0) edge (c);

\path[->]
(c) edge	node[above]	{\,\,\,\,\,\,\,${* \choose (3,0)}$} (30)
(c) edge[]	node[right, near end]	{\!${* \choose (2,0)}$} (20)
;

\path[draw, rounded corners,->]
(c) --	node[right, near start]	{${* \choose (1,0)}$} (2.8*\x,\y+\yoffset) -- (2.8*\x,3*\y-\yoffset) -- (10p);

\path[draw, rounded corners,->]
(c) --	node[right, near start]	{${* \choose (1,1)}$} (2.8*\x,-\y-\yoffset) -- (2.8*\x,-3*\y+\yoffset) -- (11p);

\path[->]
(c) edge	node[below]	{\,\,\,\,\,\,\,\,${* \choose (3,1)}$} (31)
(c) edge[]	node[right, near end]	{\!${* \choose (2,1)}$} (21)
;
\path[->]
(l4) edge node[above] {${* \choose *}$} (l3)
(l3) edge node[above] {${* \choose *}^2$} (l2)
(l2) edge node[above] {${* \choose *}^2$} (l1)
(l1) edge node[above,near start] {${* \choose *}$} (c);

\path[->]
(30) edge node[above] {${* \choose *}^2$} (30p)
(20) edge node[above] {${* \choose *}$} (20p)
(31) edge node[below] {${* \choose *}^2$} (31p)
(21) edge node[below] {${* \choose *}$} (21p);
\path[->]
(30p) edge node[right] {${0 \choose *}$} (l1)
(31p) edge node[right] {${1 \choose *}$} (l1)
(20p) edge node[right] {${0 \choose *}$} (l2)
(21p) edge node[right] {${1 \choose *}$} (l2)
(10p) edge node[right] {${0 \choose *}$} (l3)
(11p) edge node[right] {${1 \choose *}$} (l3)
(10p) edge[bend right] node[right] {${1 \choose *}$} (l4)
(11p) edge[bend left] node[right] {${0 \choose *}$} (l4);

\end{tikzpicture}
 }}
\end{floatrow}
\end{figure*}

Another simple example shows that even exponential lookahead might be necessary to achieve the smallest cost possible, relying on the exponential lower bound shown in Example~\ref{example_delaygame}.

\begin{theorem}
\label{thm_tradeoffexp}
For every $n>0$, there is a language $L_n'$ recognized by a finitary Büchi automaton with costs~$\aut_n'$ with  $\bigo(n)$ states such that 
\begin{itemize}
	\item Player~$O$ wins $\delaygame{L_n'}$ for every delay function~$f$, but
	\item an optimal strategy for $\delaygamec{2^n}{L_n'}$ has cost~$0$, and
	\item an optimal strategy for $\delaygamec{k}{L_n'}$ for $k<2^n$ has cost~$n$.
\end{itemize}
\end{theorem}

 \begin{proof}
Let $\aut_n'$ be the disjoint union of the automaton~$\aut_n$ from Example~\ref{example_automata}, where we modify the coloring to assign every state (but the sink states) color~$2$ (i.e., every run avoiding the sinks has cost~$0$) and a cycle of $n+1$ increment-transitions with exactly one state of color~$1$, every other state has color~$2$ (i.e., every run has cost~$n$). Finally, we add a fresh initial state and transitions to let Player~$O$ decide with her first move in which automaton the remaining outcome is processed. As she can always move into the cycle, which only has accepting runs, she wins  $\delaygame{L_n'}$ for every delay function~$f$. However, every strategy moving into the cycle with the first move has cost $n$. 

On the other hand, we have argued in Example~\ref{example_delaygame} that with exponential lookahead, Player~$O$ can avoid the sinks states of $\aut_n$, and thereby guarantee cost~$0$. Finally, with smaller lookahead, Player~$O$ has to enter the cycle, as we have argued in Example~\ref{example_delaygame} that Player~$I$ is able to force the run on the outcome into a sink state in that case. 
 \end{proof}

Finally, we generalize Theorem~\ref{thm_tradeoff} to a gradual tradeoff, i.e., with every additional increase of the lookahead decreases the cost of an optimal strategy, up to some upper bound. 

\begin{theorem}
\label{thm_tradeoff_gradual}
For every $n>0$, there is a language~$L_n''$ recognized by a finitary Büchi automaton with costs~$\aut_n''$ with $\bigo(n^2)$ states such that for every $j \in \set{0,1, \ldots, n}$: an optimal strategy for $\delaygamec{j}{L_n''}$ exists, but has cost~$2(n+1)-j$. 
\end{theorem}

\begin{proof}
The finitary Büchi automaton~$\aut_n''$ has alphabet~$\set{0,1} \times \set{(i,j) \mid i \in \set{1, \ldots, n}, j \in \set{0,1} }$ and is depicted in Figure~\ref{fig_tradeoff_gradual} for $n = 3$.

The automaton generalizes the idea from Theorem~\ref{thm_tradeoff}. For every $j \in \set{1, 2, 3}$, with lookahead $j$ at the initial state, Player~$O$ can use a transition of the form $(j,b)$, where $b$ is the letter picked by Player~$I$ $j$ positions ahead. The resulting run infix from the initial state back to it has a request that is answered with cost~$2(n+1)-j$, none with larger cost, and ends with all requests being answered. With less lookahead, Player~$I$ can falsify the prediction by moving from the gray states to the sink state and thereby win. Thus, an optimal strategy for $\delaygamec{j}{L_n''}$ with $j>0$ has cost~$2(n+1)-j$. 

Finally, for $j = 0$, Player~$O$ has to always pick a letter of the form~$(1,b)$ when at the initial state. The prediction can be immediately falsified by Player~$I$ by picking $1-b$ in the next round, leading to a request that is answered with cost~$8 = 2(n+1)-j$. 

The automaton~$\aut_3''$ can easily be generalized to an arbitrary $n$ by allowing Player~$O$ for every $j \in \set{1, \ldots, n}$ to predict the letter picked by Player~$I$ $j$ positions ahead with a cost of $2(n+1)-j$. 
 \end{proof}

After exhibiting these tradeoffs, a natural question concerns upper bounds on the tradeoff between quality and lookahead. The results on (delay-free, arena-based) parity games with costs imply that an optimal strategy for $\delaygamec{0}{L(\aut)}$ has cost at most $2\size{\aut}$ (with a little more effort, the factor~$2$ can be eliminated): every such game can be modeled as an arena-based parity game with costs by splitting the transitions of the automaton into two moves. For such games, it is known that the cost of an optimal strategy is at most the number of states of the arena~\cite{FijalkowZimmermann14,WeinertZimmermann16}. On the other hand, we have shown in Corollary~\ref{cor_upperboundlookaheadandbound} that exponential lookahead and exponential cost is (simulatenously) achievable, if Player~$O$ wins at all. These results constrain the type of lookahead exhibited in the previous theorems.

\subsection{Tradeoffs for Streett Conditions}
\label{subsec_tradeoffs_streett}
To conclude this section, we consider Streett conditions with costs. Recall that there is a trivial exponential lower bound on the necessary lookahead in delay games with Streett conditions with costs obtained from the same lower bound for parity conditions with costs. However, we only proved a doubly-exponential upper bound. Next, we show that this upper bound is tight, when considering strategies realizing the smallest possible cost with respect to \emph{all} delay functions.

To this end, we consider a modification of the \emph{bad $j$-pair game} described in Examples~\ref{example_automata} and \ref{example_delaygame} showing an exponential lower bound on the necessary lookahead for parity conditions with costs. Recall that Player~$O$ needs lookahead~$2^n+1$ when picking numbers from~$\set{1, \ldots, n}$. Thus, to prove a doubly-exponential lower bound, it suffices to implement this game with numbers from the range~$\set{0, \ldots, 2^n-1}$ (encoded in binary to keep the alphabet small) by an automaton of polynomial size in~$n$. However, we have to modify the rules of the game, as such a small automaton cannot recognize the winning condition, which requires to distinguish~$2^n$ different choices for $y_0$. Instead, we would like to require Player~$O$ to pick~$y_i = y_0$ for all $i>0$ and to mark the two positions~$i$ and $i'$ inducing the bad $y_0$-pair by some special markers~$\markl$ and $\markr$. Then, the automaton just has to check that $x_j$ and $y_j$ are equal at the marked positions and that $x_j$ is strictly smaller than $y_j$ in between these positions. Due to the binary encoding, both checks are easily implemented using the transition structure. 

It remains to explain how to require Player~$O$ to copy her choice~$y_0$. As before, using the state space requires too many states. Instead, we employ the finitary Streett condition with respect to a small bound to enforce the copying. The $j$-th bit of a binary encoding of a number opens a request that can only be answered by encountering the same bit at the same position of a later encoding. Thus, to answer these requests with cost~$n$, the numbers have to be copied. 

To simplify our notation, we say that a delay function~$f$ eventually grants a lookahead of size~$m$, if there is an $i$ such that $\sum_{0 \le i' \le i}(f(i)-1) \ge m$. Furthermore, we say a Streett pair in a Streett automaton is qualitative (finitary), if the associated cost function assigns $\eps$ ($\inc$) to every transition. 

\begin{theorem}
For every $n>0$, there is a language $L_n$ recognized by a Streett automaton with costs~$\aut_n$ of polynomial size in $n$ such that 
\begin{itemize}
	\item Player~$O$ has a winning strategy~$\stratO$ for $\delaygame{L_n}$ for some $f$ with $\cost(\stratO) = n$, but
	\item if Player~$O$ has a winning strategy~$\stratO$ for $\delaygame{L_n}$ with $\cost(\stratO) = n$, then $f$ eventually grants a lookahead of size~$n\cdot (2^{2^n}-1)$.
\end{itemize}
\end{theorem}

\begin{proof}
We start by describing the language~$L_n$ and by arguing that it can be recognized by a Streett automaton with costs~$\aut_n$ of polynomial size in~$n$. Fix $\SigmaI = \set{0,1} \cup \set{0,1}\times\set{{\sharpsym}}$ and $\SigmaO = \set{0,1} \cup \set{0,1} \times \set{\markl,\markr}$ and consider a word~${ \alpha \choose \beta } \in (\SigmaI \times \SigmaO)^\omega$. By grouping the bits of $\alpha$ (ignoring the mark~${\sharpsym}$) into blocks of length~$n$, $\alpha$ can be interpreted as a sequence~$\overline{\alpha} = x_0 x_1 x_2 \cdots \in \set{0, \ldots, 2^n-1}^\omega$ of natural numbers. Analogously, the bits of $\beta$ can be interpreted as a sequence~$\overline{\beta} = y_0 y_1 y_2 \cdots \in \set{0, \ldots, 2^n-1}^\omega$ when ignoring the marks~$\markl$ and $\markr$. We say that  $x_i$ is marked, if ${\sharpsym}$ holds at the first position of the block encoding~$x_i$, that $y_i$ is marked by $\markl$, if $\markl$ holds at the first position of the block encoding~$y_i$, and $y_i$ being marked by $\markr$ is defined analogously. Player~$I$ uses his mark to start a new round while Player~$O$ uses her marks to pick bad $j$-pairs.

Fix a word~$w = {\alpha \choose \beta}$ with $\overline{\alpha} = x_0 x_1 x_2 \cdots$ and $\overline{\beta} = y_0 y_1 y_2 \cdots$. If $x_0$ is not marked, then $w \in L_n$. Thus, assume from now on that $x_0$ is marked and let $i_{\sharpsym}$ be arbitrary with $x_{i_{{\sharpsym}}}$ being marked. To be in $L_n$, $w$ has to satisfy the following condition (amongst others): there have to be exactly two marked $y_i$ with $i \ge i_{\sharpsym}$ before the next $x_{i_{\sharpsym}'}$ is marked. The first one after $i_{\sharpsym}$ has to be marked by $\markl$, the second one by $\markr$. However, if there is another marked~$x_{i_{\sharpsym}'}$ before the $y_i$ marked by $\markr$ then $w$ is in $L_n$, i.e., Player~$I$ may only start a new round after Player~$O$ has picked a bad $j$-pair in the current round. These properties can be implemented using the transition structure of the automaton and a classical Streett pair to require Player~$O$ to use her marks eventually. Also, whenever a $y_i$ is marked, it has to be equal to $x_i$. Furthermore, let $y_i$ be marked by $\markl$ and let $y_{i'}$ be the next marked number. Then, we require $x_j < y_j$ for every $j$ in the range~$i < j <i'$. These requirements can be enforced by the transition structure. 

Finally, we employ finitary Streett pairs and a small bound on the cost to enforce the copying. Every occurrence of a bit~$b$ (ignoring the marks) in $\beta$ at a position~$k$ opens a request that is only answered by a later $b$ (again ignoring the marks) in $\beta$ at a position~$k'$ with $k \bmod n = k' \bmod n$ or by a later block that is marked by $\markr$. This property is enforced by finitary Streett pairs. Thus, to answer these requests with cost~$n$, all the $y_j$ between a marked~$x_i$ (the start of a round) and the next $y_{i'}$ marked by $\markr$ (the end of the round) have to coincide.

It is straightforward to construct a Streett automaton with costs~$\aut_n$ of polynomial size that recognizes the language~$L_n$ described above. 

Now, analogously to the arguments in Example~\ref{example_delaygame}, one can show that Player~$O$ wins $\delaygame{L_n}$ for the constant~$f$ with $f(0) = n \cdot 2^{2^n}$: at the start of each round, she has enough lookahead to pick a $y_i$ such that Player~$I$ has already produced a bad $y_i$-pair in the current round. Then, she copies the $y_i$ and marks the pair correctly and waits for the start of the next round. This satisfies the qualitative Streett pairs as well as the finitary ones with cost~$n$, i.e, she wins.

Now fix some delay function~$f$ such that Player~$O$ has a winning strategy~$\stratO$ for $\delaygame{L_n}$ with $\cost(\stratO) = n$. Assume towards a contradiction that $f$ does not eventually grant a lookahead of size~$n\cdot (2^{2^n}-1)$. As mentioned in Example~\ref{example_delaygame}, there is a sequence~$\overline{x} $ of length~$2^{2^n}-1$ without a bad $j$-pair for every $j \in \set{0, \ldots, 2^n-1}$. Player~$I$'s strategy against $\stratO$ is to first play the binary encoding of $\overline{x}$, with the first number marked. Due to the small lookahead, Player~$O$ has to specify $y_0$ during these moves. Then, Player~$I$ just plays some $x \neq y_0$ until Player~$O$ has played both her marks. If she never does, then Player~$I$ wins the play. Otherwise, he just starts a new round and proceeds as previously described. 

Consider a round of an outcome of this strategy and say Player~$O$ picks $y_i$ as first number in this round. Then, the sequence of numbers picked by Player~$I$ in this round contains no bad $y_i$-pair. If Player~$O$ does not mark any numbers in this round, or they do not constitute a bad $j$-pair for some $j$, then she loses the play, which contradicts our assumption. Thus, assume she does mark a bad $j$-pair correctly. Then, we have $j \neq y_0$. This means she does not copy $y_i$ throughout the round until she plays her second mark. Thus, there is a request that is not answered with cost~$n$. As Player~$I$ is able to enforce such a request in each round, the resulting play has at least cost $n+1$, again a contradiction. 
\end{proof}

On the other hand, the previous example does not yield a doubly-exponential lower bound on the necessary lookahead for arbitrary bounds on the costs, as Player~$O$ can satisfy the acceptance condition of the automaton by starting each round by playing~$0^n$ and then $1^n$, which answers all finitary Streett pairs in the game and then correctly marking a bad $j$-pair for some $j$. This strategy has much larger cost than $n$, but is still winning. Whether there is a tradeoff (and, if yes, it's extent) remains an open problem. 

Also note that the automaton~$\aut_n$ has both classical and finitary Streett pairs. It is an open problem to show the same result for finitary Streett automata. The problem one encounters  is that the acceptance condition has to force Player~$O$ to mark a bad $j$-pair in each round. Implementing this with a finitary Streett pair increases the cost of a winning strategy, as it may take doubly-exponentially long before such a pair appears. This large bound allows Player~$O$ to cheat in the copying process, as described above.

\section{Conclusion}
\label{sec_conc}
We have demonstrated the usefulness of adding delay to games with quantitative winning conditions, here finitary parity and Streett conditions as well as parity and Streett conditions with costs. 

We have shown that delay games with parity conditions with costs are just as hard as delay games with parity conditions, both in terms of the necessary lookahead and in terms of the computational complexity of determining the winner. Thus, adding quantitative features to such games comes for free, which is in line with similar results for both delay-free games~\cite{FijalkowZimmermann14,KupfermanPitermanVardi09,Zimmermann13} and delay games~\cite{KleinZimmermann16b}. Furthermore, we exhibited the usefulness of delay by showing that lookahead can be traded for quality of strategies. This phenomenon goes beyond the advantages in qualitative delay games, where lookahead only allows to win more games. 

Another interesting property of delay-free finitary parity games is that playing them optimally is much harder than just winning them: the bounding player always has a positional winning strategy~\cite{FijalkowZimmermann14}, which is winning with respect to some uniform bound~$b$, but satisfying the optimal uniform bound might require exponential memory~\cite{WeinertZimmermann16}. Similarly, determining the optimal bound is $\pspace$-complete~\cite{WeinertZimmermann16} while just determining the winner of a delay-free finitary parity game is in $\ptime$~\cite{ChatterjeeHenzingerHorn09}. 

In current work, we study the tradeoffs between quality, memory requirements, lookahead, and solution complexity. In particular, this requires to develop a theory of finite-state strategies for delay games, which is, due to the presence of lookahead, non-trivial (see~\cite{Salzmann15} for a proposal of finite-state strategies in a setting that is similar to the definition of $\game(\aut)$). 

Furthermore, we gave a doubly-exponential upper bound on the necessary lookahead for delay games with Streett conditions with costs and showed that such games can be solved in doubly-exponential time. The best lower bounds are those for parity conditions with costs, i.e., there is an exponential gap in both cases. Note that these gaps already exists in the case of qualitative Streett conditions: doubly-exponential constant lookahead is sufficient and solving such games is in $\twoexp$ (this follows via determinization from the results for parity conditions), but the best lower bounds are exponential for the lookahead and $\exptime$-completeness. In future work, we aim to close these gaps. 

\paragraph*{Acknowledgments} We thank Alexander Weinert for numerous fruitful discussions. 

\bibliographystyle{splncs03}
\bibliography{main.bib}

\end{document}